\documentclass[journal,twocolumn,twoside]{IEEEtran}

\usepackage{amsfonts}
\usepackage{color}
\usepackage{graphicx}
\usepackage[dvips]{epsfig}
\usepackage{graphics} 
\usepackage{times} 
\usepackage[cmex10]{amsmath} 
\usepackage{amssymb}  
\usepackage{cite}
\usepackage{mathtools}
\usepackage{multirow}
\usepackage[tight,footnotesize]{subfigure}
\usepackage{algorithm}
\usepackage{amsthm}

\def\norm #1{\left\|#1\right\|}
\def\inftyn #1{\left\|#1\right\|_{\infty}}
\def\twon #1{\left\|#1\right\|_2}
\def\onen #1{\left\|#1\right\|_1}

\def\frobn #1{\left\|#1\right\|_{\text{F}}}

\def\twoinfn #1{\left\|#1\right\|_{2,\infty}}
\def\atomn #1{\left\|#1\right\|_{\cA}}

\def\abs #1{\left|#1\right|}

\def\st{\text{subject to }}

\def\bC{\mathbb{C}}

\def\bR{\mathbb{R}}
\def\bS{\mathbb{S}}
\def\bT{\mathbb{T}}
\def\bE{\mathbb{E}}
\def\bP{\mathbb{P}}

\def\m #1{\boldsymbol{#1}}

\def\cA{\mathcal{A}}

\def\cE{\mathcal{E}}

\def\cK{\mathcal{K}}

\def\cP{\mathcal{P}}
\def\cT{\mathcal{T}}

\def\cT{\mathcal{T}}

\def\bee{\begin{equation}}
\def\ene{\end{equation}}

\def\beq{\begin{eqnarray}}
\def\enq{\end{eqnarray}}
\def\lentwo{\setlength\arraycolsep{2pt}}

\newtheorem{lem}{Lemma}

\newtheorem{thm}{Theorem}

\newtheorem{defi}{Definition}

\def\equ #1{\begin{equation}#1\end{equation}}
\def\equa #1{\begin{eqnarray}#1\end{eqnarray}}
\def\sbra #1{\left(#1\right)}
\def\mbra #1{\left[#1\right]}
\def\lbra #1{\left\{#1\right\}}
\def\diag #1{\text{diag}#1}
\def\tr #1{\text{tr}#1}

\def\rank #1{\text{rank}#1}
\def\spark #1{\text{spark}#1}
\def\st {\text{ subject to }}

\title{Exact Joint Sparse Frequency Recovery via Optimization Methods}
\author{Zai Yang, {\em Member, IEEE}, and Lihua Xie, {\em Fellow, IEEE}
\thanks{This work appeared in part in the {\em Proceedings of the 2014 IEEE Workshop on Statistical Signal Processing (SSP)}, Gold Coast, Australia, June 2014 \cite{yang2014continuous}. The research of the project was supported by Ministry of Education, Republic of Singapore, under grant AcRF TIER 1 RG78/15.

Z. Yang is with the School of Automation, Nanjing University of Science and Technology, Nanjing 210094, China, and also with the School of Electrical and Electronic Engineering, Nanyang Technological University, Singapore 639798 (e-mail: yangzai@ntu.edu.sg).

L. Xie is with the School of Electrical and Electronic Engineering, Nanyang Technological University, Singapore 639798 (e-mail: elhxie@ntu.edu.sg).}}

\begin{document}
\maketitle


\begin{abstract} Frequency recovery/estimation from discrete samples of superimposed sinusoidal signals is a classic yet important problem in statistical signal processing. Its research has recently been advanced by atomic norm techniques which exploit signal sparsity, work directly on continuous frequencies, and completely resolve the grid mismatch problem of previous compressed sensing methods. In this work we investigate the frequency recovery problem in the presence of multiple measurement vectors (MMVs) which share the same frequency components, termed as joint sparse frequency recovery and arising naturally from array processing applications. To study the advantage of MMVs, we first propose an $\ell_{2,0}$ norm like approach by exploiting joint sparsity and show that the number of recoverable frequencies can be increased except in a trivial case. While the resulting optimization problem is shown to be rank minimization that cannot be practically solved, we then propose an MMV atomic norm approach that is a convex relaxation and can be viewed as a continuous counterpart of the $\ell_{2,1}$ norm method. We show that this MMV atomic norm approach can be solved by semidefinite programming. We also provide theoretical results showing that the frequencies can be exactly recovered under appropriate conditions. The above results either extend the MMV compressed sensing results from the discrete to the continuous setting or extend the recent super-resolution and continuous compressed sensing framework from the single to the multiple measurement vectors case. Extensive simulation results are provided to validate our theoretical findings and they also imply that the proposed MMV atomic norm approach can improve the performance in terms of reduced number of required measurements and/or relaxed frequency separation condition.

\end{abstract}




\section{Introduction}

Suppose that we observe uniform samples (with the Nyquist sampling rate) of a number of $L$ sinusoidal signals:
\equ{y_{jt}^o=\sum_{k=1}^K s_{kt}e^{i2\pi j f_k}, \quad \sbra{j,t}\in \m{J}\times\mbra{L}, \label{formu:observmodel1}}
which form an $N\times L$ matrix $\m{Y}^o=\mbra{y_{jt}^o}$, on the index set $\m{\Omega}\times\mbra{L}$, where $\m{\Omega} \subset \m{J} \coloneqq \lbra{0,1,\dots,N-1}$, $\mbra{L}\coloneqq \lbra{1,2,\dots,L}$, and $N$ is the number of uniform samples per sinusoidal signal. This means that each sinusoidal signal corresponds to one measurement vector.
Here $(j,t)$ indexes the entries of $\m{Y}^o$, $i=\sqrt{-1}$, $f_k\in\bT\coloneqq \left[0,1\right]$ denotes the $k$th normalized frequency (note that the starting point 0 and the ending point 1 of the unit circle $\bT$ are identical), $s_{kt}\in\bC$ is the (complex) amplitude of the $k$th frequency component composing the $t$th sinusoidal signal, and $K$ is the number of the components which is small but unknown. Moreover, let $M=\abs{\m{\Omega}}\leq N$ be the sample size of each measurement vector. The observed $M\times L$ data matrix $\m{Y}_{\m{\Omega}}^o\coloneqq\lbra{y_{jt}^o}_{\sbra{j,t}\in\m{\Omega}\times \mbra{L}}$ are referred to as full data, if $M=N$ (i.e., $\m{\Omega} = \m{J}$ and $\m{Y}_{\m{\Omega}}^o=\m{Y}^o$), and otherwise, compressive data. Let $\cT=\lbra{f_1,\dots,f_K}$ denote the set of the frequencies. The problem concerned in this paper is to recover $\cT$ given $\m{Y}_{\m{\Omega}}^o$, which is referred to as joint sparse frequency recovery (JSFR) in the sense that the multiple measurement vectors (MMVs) (i.e., the $L$ columns of $\m{Y}_{\m{\Omega}}^o$) share the same $K$ frequencies. Once $\cT$ is obtained, the amplitudes $\lbra{s_{kt}}$ and the full data $\m{Y}^o$ can be easily obtained by a simple least-squares method.

An application of the JSFR problem is direction of arrival (DOA) estimation in array processing \cite{krim1996two,stoica2005spectral}. In particular, suppose that $K$ farfield, narrowband sources impinge on a linear array of sensors and one wants to know their directions. The output of the sensor array can be modeled by (\ref{formu:observmodel1}) under appropriate conditions, where each frequency corresponds to one source's direction. The sampling index set $\m{\Omega}$ therein represents the geometry of the sensor array. To be specific, $\m{\Omega}=\m{J}$ in the full data case corresponds to an $N$-element uniform linear array (ULA) with adjacent sensors spaced by half a wavelength, while $\m{\Omega}\subsetneq\m{J}$ corresponds to a sparse linear array (SLA) that can be obtained by retaining only the sensors of the above ULA indexed by $\m{\Omega}$. Each measurement vector consists of the outputs of the sensor array at one snapshot. The $L$ MMVs are obtained by taking $L$ snapshots under the assumption of static sources (during a time window). Note that, since the array size can be limited in practice due to physical constraints and/or cost considerations, it is crucial in DOA estimation to exploit the temporal redundancy (a.k.a., the joint sparsity that we refer to) contained in the MMVs.

In conventional methods for JSFR one usually assumes that the source signals (or the rows of $\mbra{s_{kt}}$) have zero mean and are spatially uncorrelated. It follows that the covariance matrix of the full data snapshot (or the columns of $\m{Y}^o$) is positive semidefinite (PSD), Toeplitz and low rank (of rank $K$). Exploiting these structures for frequency recovery was firstly proposed by Pisarenko who rediscovered the classical Vandermonde decomposition lemma that states that the frequencies can be exactly retrieved from the data covariance matrix \cite{caratheodory1911zusammenhang,pisarenko1973retrieval}. A prominent class of methods was then proposed and designated as subspace methods such as MUSIC and ESPRIT \cite{schmidt1981signal,roy1989esprit}. While these methods estimate the data covariance using the sample covariance, the Toeplitz structure cannot be exploited in general, a sufficient number of snapshots is required, and their performance can be degraded in the presence of source correlations.

With the development of sparse signal representation and later the compressed sensing (CS) concept \cite{candes2006robust, donoho2006compressed}, sparse (for $L=1$) and joint sparse (for $L>1$) methods for frequency recovery have been popular in the past decade. In these methods, however, the frequencies of interest are usually assumed to lie on a fixed grid on $\bT$ because the development of CS so far has been focused on signals that can be sparsely represented under a finite discrete dictionary. Under the on-grid assumption, the observation model in (\ref{formu:observmodel1}) can be written into an underdetermined system of linear equations and CS methods are applied to solve an involved sparse signal whose support is finally identified as the frequency set $\cT$. Typical sparse methods include combinatorial optimization or $\ell_0$ (pseudo-)norm minimization, its convex relaxation or $\ell_1$ norm minimization, and greedy methods such as orthogonal matching pursuit (OMP) as well as their joint sparse counterparts \cite{donoho2003optimally,tropp2007signal,malioutov2005sparse,chen2006theoretical,hyder2010direction, eldar2010average}. While the $\ell_0$ minimization can exploit sparsity to the greatest extent possible, it is NP-hard and cannot be practically solved. The maximal $K$ allowed in $\ell_1$ minimization and OMP for guaranteed exact recovery is inversely proportional to a metric called coherence which, however, increases dramatically as the grid gets fine. Moveover, grid mismatches become a major problem of CS-based methods though several modifications have been proposed to alleviate this drawback (see, e.g., \cite{hu2012compressed,yang2012robustly, yang2013off,austin2013dynamic}).

Breakthroughs came out recently. In the single measurement vector (SMV) case when $L=1$, Cand\`{e}s and Fernandez-Granda \cite{candes2013towards} dealt directly with continuous frequencies and completely resolved the grid mismatch problem. In particular, they considered the full data case and showed that the frequencies can be exactly recovered by exploiting signal sparsity if all the frequencies are mutually separated by at least $\frac{4}{N}$. This means that up to $K=\frac{N}{4}$ frequencies can be recovered. Their method is based on the total variation norm or the atomic norm that extends the $\ell_1$ norm from the discrete to the continuous frequency case and can be computed using semidefinite programming (SDP) \cite{rudin1987real,chandrasekaran2012convex}. Following from \cite{candes2013towards}, Tang {\em et al.} \cite{tang2012compressed} studied the same problem in the case of compressive data using atomic norm minimization (ANM). Under the same frequency separation condition, they showed that a number of $M\geq O\sbra{K\log K\log N}$ randomly selected samples is sufficient to guarantee exact recovery with high probability. Several subsequent papers on this topic include \cite{bhaskar2013atomic,candes2013super,tang2013near,fang2014super,azais2015spike,yang2015gridless}. However, similar {\em gridless sparse} methods are rare for JSFR in the MMV case concerned in this paper. A gridless method designated as the sparse and parametric approach (SPA) was proposed in our previous work \cite{yang2014discretization} based on weighted covariance fitting by exploiting the structures of the data covariance matrix. In the main context of this paper we will show that this method is closely related to the MMV atomic norm method that we will introduce in the present paper. Another related work is \cite{tan2014direction}; however, in this paper the MMV problem was reformulated as an SMV one, with the joint sparsity missing, and solved within the framework in \cite{candes2013towards}. Therefore, the frequency recovery performance can be degraded. As an example, in the noiseless case the frequencies cannot be exactly recovered using the method in \cite{tan2014direction} due to some new `noise' term introduced.

In this paper, we first study the advantage of exploiting joint sparsity in the MMVs and then propose a practical approach to utilize this information. In particular, following from the literature on CS we propose an $\ell_0$ norm like sparse metric that is referred to as the MMV atomic $\ell_0$ norm and is a continuous counterpart of the $\ell_{2,0}$ norm used for joint sparse recovery \cite{chen2006theoretical}. We theoretically show that the MMVs can help improve the frequency recovery performance in terms of the number of recoverable frequencies except in a trivial case. But unfortunately (in fact, not surprisingly), this atomic $\ell_0$ norm approach is proven to be a rank minimization problem that cannot be practically solved. We then propose a convex relaxation approach in which the MMV atomic norm is adopted that is a continuous counterpart of the $\ell_{2,1}$ norm. We show that this atomic norm approach can be efficiently solved via semidefinite programming. Theoretical results are also provided to show that the frequencies can be exactly recovered under similar conditions as in \cite{candes2013towards, tang2012compressed}. Extensive simulation results are provided to validate our theoretical results and they also imply that the proposed MMV atomic norm approach can result in improved frequency recovery performance in terms of reduced number of required measurements and/or relaxed frequency separation condition.

It is interesting to note that the proposed MMV atomic $\ell_0$ norm and atomic norm approaches somehow exploit the structures of the ``data covariance matrix'' and are related to the aforementioned subspace methods. In particular, a PSD Toeplitz matrix is involved in both the proposed methods that can be interpreted as the data covariance matrix (as if certain statistical assumptions were satisfied) from the Vandermonde decomposition of which the true frequencies are finally obtained, while the low rank structure is exploited by matrix rank minimization in the atomic $\ell_0$ norm method and by matrix trace norm (or nuclear norm) minimization in the atomic norm method. As compared to the subspace methods, the proposed methods exploit the matrix structures to a greater extent. Moreover, the proposed methods do not require the assumption of uncorrelated sources and can be applied to the case of limited measurement vectors.

The results of this work were published online in the technical report \cite{yang2014exact1} and were presented in part in the conference paper \cite{yang2014continuous}. When preparing this paper we found that the same MMV atomic norm approach was also independently proposed in \cite{chi2014joint,li2014off}. This paper is different form \cite{chi2014joint,li2014off} in the following aspects. First, in this paper, the advantage of MMVs is theoretically proven in terms of the number of recoverable frequencies based on the proposed MMV atomic $\ell_0$ norm approach, while no such theoretical results are provided in \cite{chi2014joint,li2014off}. Second, in this paper, the SDP formulation of the MMV atomic norm is proven inspired by our previous work \cite{yang2014discretization}, while the proof in \cite{chi2014joint,li2014off} is given following \cite{tang2012compressed} on the SMV case. Finally, as pointed out in \cite{li2014off}, the theoretical guarantee of the MMV atomic norm approach provided in \cite[Theorem 2]{li2014off} is weaker than ours (see Theorem \ref{thm:incompletedata}; note that the technical report \cite{yang2014exact1} appeared online earlier than \cite{li2014off}).

Notations used in this paper are as follows. $\bR$ and $\bC$ denote the sets of real and complex numbers respectively. $\bT$ denotes the unit circle $\left[0,1\right]$ by identifying the starting and ending points. Boldface letters are reserved for vectors and matrices. For an integer $L$, $[L]\coloneqq\lbra{1,\cdots,L}$. $\abs{\cdot}$ denotes the amplitude of a scalar or the cardinality of a set. $\onen{\cdot}$, $\twon{\cdot}$ and $\frobn{\cdot}$ denote the
$\ell_1$, $\ell_2$ and Frobenius norms respectively. $\m{A}^T$ and $\m{A}^H$ are the matrix transpose and conjugate transpose of $\m{A}$ respectively. $x_j$ is the $j$th entry of a vector $\m{x}$, and $\m{A}_j$ denotes the $j$th row of a matrix $\m{A}$. Unless otherwise stated, $\m{x}_{\m{\Omega}}$ and $\m{A}_{\m{\Omega}}$ are subvector and submatrix of $\m{x}$ and $\m{A}$ respectively by retaining the entries of $\m{x}$ and the rows of $\m{A}$ indexed by the set $\m{\Omega}$. For a vector $\m{x}$, $\diag\sbra{\m{x}}$ is a diagonal matrix with $\m{x}$ on the diagonal. $\m{x}\succeq\m{0}$ means $x_j\geq0$ for all $j$. $\rank\sbra{\m{A}}$ denotes the rank of a matrix $\m{A}$ and $\tr\sbra{\m{A}}$ the trace. For positive semidefinite matrices $\m{A}$ and $\m{B}$, $\m{A}\geq\m{B}$ means that $\m{A}-\m{B}$ is positive semidefinite. $\bE\mbra{\cdot}$ denotes the expectation and $\bP\sbra{\cdot}$ the probability of an event.

The rest of the paper is organized as follows. Section \ref{sec:result} presents the main results of this paper. Section \ref{sec:connection} discusses connections between the proposed methods and prior art. Section \ref{sec:proof} presents proofs of the main results in Section \ref{sec:result}. Section \ref{sec:simulation} provides numerical simulations and Section \ref{sec:conclusion} concludes this paper.

\section{Main Results} \label{sec:result}

This section presents the main results of this paper whose proofs  will be given in Section \ref{sec:proof}.

\subsection{Preliminary: Vandermonde Decomposition}
The Vandermonde decomposition of Toeplitz matrices can date back to 1910s and has been important in the signal processing society since its rediscovery and use for frequency estimation in 1970s \cite{caratheodory1911zusammenhang,pisarenko1973retrieval} (see also \cite{stoica2005spectral}). In particular, it states that any PSD, rank-$K\leq N$, Toeplitz matrix $T\sbra{\m{u}}\in\bC^{N\times N}$, which is parameterized by $\m{u}\in\bC^N$ and given by
\equ{T\sbra{\m{u}}=\begin{bmatrix}u_1 & u_2 & \cdots & u_N\\ {u}_2^H & u_1 & \cdots & u_{N-1}\\ \vdots & \vdots & \ddots & \vdots \\ {u}_N^H & {u}_{N-1}^H & \cdots & u_1\end{bmatrix}, }
can be decomposed as
\equ{T\sbra{\m{u}}=\sum_{k=1}^K p_k \m{a}\sbra{f_k}\m{a}^H\sbra{f_k} = \m{A}\sbra{\m{f}}\m{P} \m{A}^H\sbra{\m{f}},\label{formu:VD}}
where $\m{A}\sbra{\m{f}}=\mbra{\m{a}\sbra{f_1},\dots,\m{a}\sbra{f_K}}\in\bC^{N\times K}$ with $\m{a}\sbra{f}=\mbra{1,e^{i2\pi f},\dots,e^{i2\pi\sbra{N-1}f}}^T\in\bC^N$, $\m{P}=\diag\sbra{p_1,\dots,p_K}$ with $p_k>0$, $k=1,\dots,K$ and $\lbra{f_k}$ are distinct points in $\bT$. Moreover, the decomposition in \eqref{formu:VD} is unique if $K<N$. Note that the name `Vandermonde' comes from the fact that $\m{A}\sbra{\m{f}}$ is a Vandermonde matrix.

It is well known that under the assumption of uncorrelated sources the data covariance matrix (i.e., the covariance matrix of each column of $\m{Y}^o$) is a rank-$K$, PSD, Toeplitz matrix. Therefore, the Vandermonde decomposition actually says that the frequencies can be uniquely obtained from the data covariance matrix given $K<N$ \cite{pisarenko1973retrieval}. Note that a subspace method such as ESPRIT can be used to compute the decomposition in \eqref{formu:VD}.

\subsection{Frequency Recovery Using Joint Sparsity} \label{sec:atomell0min}

To exploit the joint sparsity in the MMVs, we let $\m{s}_{k}=\mbra{s_{k1},\cdots,s_{kL}}\in\bC^{1\times L}$. It follows that (\ref{formu:observmodel1}) can be written as
\equ{\m{Y}^o=\sum_{k=1}^K \m{a}\sbra{f_k}\m{s}_{k}=\sum_{k=1}^K c_k\m{a}\sbra{f_k}\m{\phi}_{k}, \label{formu:observmodel}}
where $\m{a}\sbra{f}$ is as defined in \eqref{formu:VD}, $c_k=\twon{\m{s}_{k}}>0$ and $\m{\phi}_{k}=c_k^{-1}\m{s}_{k}$ with $\twon{\m{\phi}_{k}}=1$. Let $\bS^{2L-1}=\lbra{ \m{\phi}\in\bC^{1\times L}:\; \twon{\m{\phi}}=1}$ denote the unit complex $\sbra{L-1}$-sphere (or real $\sbra{2L-1}$-sphere). Define the set of atoms
\equ{\cA\coloneqq \lbra{\m{a}\sbra{f,\m{\phi}}=\m{a}\sbra{f}\m{\phi}: f\in\bT, \m{\phi}\in\bS^{2L-1}}.\label{formu:atomset}}
It follows from (\ref{formu:observmodel}) that $\m{Y}^o$ is a linear combination of $K$ atoms in $\cA$. In particular, we say that a decomposition of $\m{Y}^o$ as in (\ref{formu:observmodel}) is an atomic decomposition of order $K$ if $c_k>0$ and the frequencies $f_k$ are distinct.

Following from the literature on CS, we first propose an (MMV) atomic $\ell_0$ norm approach to signal and frequency recovery that exploits sparsity to the greatest extent possible. In particular, the atomic $\ell_0$ norm of $\m{Y}\in\bC^{N\times L}$ is defined as the smallest number of atoms in $\cA$ that can express $\m{Y}$:
\equ{\norm{\m{Y}}_{\cA,0}
=\inf\lbra{\cK: \m{Y}=\sum_{k=1}^{\cK} c_k\m{a}_k, \m{a}_k\in\cA, c_k>0}. \label{formu:AL0}}
The following optimization method is proposed for signal recovery that generalizes a method in \cite{tang2012compressed} from the SMV to the MMV case:
\equ{\min_{\m{Y}} \norm{\m{Y}}_{\cA,0}, \st \m{Y}_{\m{\Omega}}=\m{Y}^o_{\m{\Omega}}. \label{formu:AL0min}}
The frequencies composing the solution of $\m{Y}$ are the frequency estimates.

To show the advantage of MMVs, we define the continuous dictionary
\equ{\begin{split}\cA_{\m{\Omega}}^1
&\coloneqq \lbra{\m{a}_{\m{\Omega}}\sbra{f}: f\in\bT} \end{split}}
and then define the spark of $\cA_{\m{\Omega}}^1$, denoted by $\spark\sbra{\cA_{\m{\Omega}}^1}$, as the smallest number of atoms in $\cA_{\m{\Omega}}^1$ that are linearly dependent. Note that this definition of spark generalizes that in \cite{kruskal1977three} from the discrete to the continuous dictionary case. We have the following theoretical guarantee for (\ref{formu:AL0min}).

\begin{thm} $\m{Y}^o=\sum_{j=1}^K c_j\m{a}\sbra{f_j,\m{\phi}_j}$ is the unique optimizer to \eqref{formu:AL0min} if
\equ{K< \frac{\spark\sbra{\cA_{\m{\Omega}}^1}-1+\rank \sbra{\m{Y}_{\m{\Omega}}^o}}{2}. \label{Kbound}}
Moreover, the atomic decomposition above is the unique one satisfying that $K=\norm{\m{Y}^o}_{\cA,0}$.
\label{thm:AL0_guanrantee}
\end{thm}

By Theorem \ref{thm:AL0_guanrantee} the frequencies can be exactly recovered using the atomic $\ell_0$ norm approach if the sparsity $K$ is sufficiently small with respect to the sampling index set $\m{\Omega}$ and the observed data $\m{Y}_{\m{\Omega}}^o$. Note that the number of recoverable frequencies can be increased, as compared to the SMV case, if $\rank \sbra{\m{Y}_{\m{\Omega}}^o}>1$, which happens except in a trivial case when the MMVs in $\m{Y}_{\m{\Omega}}^o$ are identical up to scaling factors.

But unfortunately, the following result shows that $\norm{\m{Y}}_{\cA,0}$ is substantially a rank minimization problem that cannot be practically solved.

\begin{thm} $\norm{\m{Y}}_{\cA,0}$ defined in (\ref{formu:AL0}) equals the optimal value of the following rank minimization problem:
\equ{\min_{\m{W},\m{u}} \rank\sbra{T\sbra{\m{u}}}, \st \begin{bmatrix}\m{W} & \m{Y}^H \\ \m{Y} & T\sbra{\m{u}}\end{bmatrix} \geq\m{0}. \label{formu:AL0_rankmin}}
\label{thm:AL0_rankmin}
\end{thm}

It immediately follows from \eqref{formu:AL0_rankmin} that \eqref{formu:AL0min} can be cast as the following rank minimization problem:
\equ{\begin{split}
&\min_{\m{Y},\m{W},\m{u}} \rank\sbra{T\sbra{\m{u}}},\\
&\st \begin{bmatrix}\m{W} & \m{Y}^H \\ \m{Y} & T\sbra{\m{u}}\end{bmatrix}\geq\m{0} \text{ and } \m{Y}_{\m{\Omega}}=\m{Y}^o_{\m{\Omega}}. \end{split} \label{formu:AL0_rankmin1}}
Suppose that \eqref{formu:AL0_rankmin1} can be globally solved and let $\m{u}^*$ and $\m{Y}^*$ denote the solutions of $\m{u}$ and $\m{Y}$, respectively. If the condition of Theorem \ref{thm:AL0_guanrantee} is satisfied, then $\m{Y}^o=\m{Y}^*$ and the frequencies as well as the atomic decomposition of $\m{Y}^o$ in Theorem \ref{thm:AL0_guanrantee} can be computed accordingly. In particular, it is guaranteed that $\rank\sbra{T\sbra{\m{u}^*}}=K<N$ (see the proof in Section \ref{sec:proofAL0}). It follows that the true frequencies can be uniquely obtained from the Vandermonde decomposition of $T\sbra{\m{u}^*}$. After that, the atomic decomposition of $\m{Y}^o$ can be obtained by the fact that $\m{Y}^*$ lies in the range space of $T\sbra{\m{u}^*}$. Moreover, it is worth noting that, although \eqref{formu:AL0min} has a trivial solution in the full data case, the problem in \eqref{formu:AL0_rankmin1} still makes sense and the frequency retrieval process also applies.


\subsection{Frequency Recovery via Convex Relaxation} \label{sec:convexrelaxation}

While the rank minimization problem in \eqref{formu:AL0_rankmin1} is nonconvex and cannot be globally solved with a practical algorithm, it motivates the (MMV) atomic norm method---a convex relaxation. In particular, the atomic norm of $\m{Y}\in\bC^{N\times L}$ is defined as the gauge function of $\text{conv}\sbra{\cA}$, the convex hull of $\cA$ \cite{chandrasekaran2012convex}:
\equ{\begin{split}\atomn{\m{Y}}
&\coloneqq\inf\lbra{t>0: \m{Y}\in t\text{conv}\sbra{\cA}} \\
&= \inf\lbra{\sum_k c_k: \m{Y}=\sum_k c_k\m{a}_k, c_k>0, \m{a}_k\in\cA},\end{split} \label{formu:atomicnorm}}
in which the joint sparsity is exploited in a different manner.
Indeed, $\atomn{\cdot}$ is a norm by the property of the gauge function and thus it is convex. Corresponding to \eqref{formu:AL0min}, we propose the following convex optimization problem:
\equ{\min_{\m{Y}} \norm{\m{Y}}_{\cA}, \st \m{Y}_{\m{\Omega}}=\m{Y}^o_{\m{\Omega}}. \label{formu:ANmin}}

Though we know that \eqref{formu:ANmin} is convex, \eqref{formu:ANmin} still cannot be practically solved since by \eqref{formu:atomicnorm} it is a semi-infinite program with an infinite number of variables. To practically solve \eqref{formu:ANmin}, an SDP formulation of $\norm{\m{Y}}_{\cA}$ is provided in the following theorem.

\begin{thm} $\norm{\m{Y}}_{\cA}$ defined in (\ref{formu:atomicnorm}) equals the optimal value of the following SDP:
\equ{\begin{split}
&\min_{\m{W},\m{u}} \frac{1}{2\sqrt{N}}\mbra{\tr\sbra{\m{W}} + \tr\sbra{T\sbra{\m{u}}}}, \\
&\st \begin{bmatrix}\m{W} & \m{Y}^H \\ \m{Y} & T\sbra{\m{u}}\end{bmatrix}\geq\m{0}.\end{split} \label{formu:AN_SDP}} \label{thm:AN_SDP}
\end{thm}

By Theorem \ref{thm:AN_SDP}, \eqref{formu:ANmin} can be cast as the following SDP which can be solved using an off-the-shelf SDP solver:
\equ{\begin{split}
&\min_{\m{Y},\m{W},\m{u}} \tr\sbra{\m{W}} + \tr\sbra{T\sbra{\m{u}}},\\
&\st \begin{bmatrix}\m{W} & \m{Y}^H \\ \m{Y} & T\sbra{\m{u}}\end{bmatrix}\geq\m{0} \text{ and } \m{Y}_{\m{\Omega}}=\m{Y}^o_{\m{\Omega}}. \end{split} \label{formu:AN_sdp}}
Given the optimal solution $\m{u}^*$ to \eqref{formu:AN_sdp}, the frequencies and the atomic decomposition of $\m{Y}^o$ can be computed as previously based on the Vandermonde decomposition of $T\sbra{\m{u}^*}$.

Finally, we analyze the theoretical performance of the atomic norm approach. To do so, we define the minimum separation of a finite subset $\cT\subset\bT$ as the closest wrap-around distance between any two elements,
\equ{\Delta_{\cT}=\inf_{a,b\in\cT:a\neq b}\min\lbra{\abs{a-b}, 1-\abs{a-b}}. \notag}
We first study the full data case that, as we will see, forms the basis of the compressive data case. Note that \eqref{formu:AN_sdp} can be solved for frequency recovery though \eqref{formu:ANmin} admits a trivial solution. We have the following theoretical guarantee.

\begin{thm} $\m{Y}^o=\sum_{j=1}^K c_j\m{a}\sbra{f_j,\m{\phi}_j}$ is the unique atomic decomposition satisfying that $\norm{\m{Y}^o}_{\cA}=\sum_{j=1}^K c_j$ if $\Delta_{\cT}\geq \frac{1}{\lfloor(N-1)/4\rfloor}$ and $N\geq257$.\footnote{The condition $N\geq257$ is more like a technical requirement but not an obstacle in practice (see numerical simulations in Section \ref{sec:simulation}).} \label{thm:completedata}
\end{thm}

In the compressive data case, the following result holds.

\begin{thm} Suppose we observe $\m{Y}^o=\sum_{j=1}^K c_j\m{a}\sbra{f_j,\m{\phi}_j}$
on the index set $\m{\Omega}\times\mbra{L}$, where $\m{\Omega}\subset\m{J}$ is of size $M$ and selected uniformly at random. Assume that $\lbra{\m{\phi}_j}_{j=1}^K\subset \bS^{2L-1}$ are  independent random variables with $\bE\m{\phi}_j=\m{0}$. If $\Delta_{\cT}\geq \frac{1}{\lfloor(N-1)/4\rfloor}$, then there exists a numerical constant $C$ such that
\equ{M\geq C\max\lbra{\log^2\frac{\sqrt{L}N}{\delta}, K\log\frac{K}{\delta}\log\frac{\sqrt{L}N}{\delta}} \label{formu:AN_bound}}
is sufficient to guarantee that, with probability at least $1-\delta$, $\m{Y}^o$ is the unique optimizer to (\ref{formu:ANmin}) and $\m{Y}^o=\sum_{j=1}^K c_j\m{a}\sbra{f_j,\m{\phi}_j}$ is the unique atomic decomposition satisfying that $\norm{\m{Y}^o}_{\cA}=\sum_{j=1}^K c_j$. \label{thm:incompletedata}
\end{thm}

\subsection{Discussions}

We have proposed two optimization approaches to JSFR by exploiting the joint sparsity in the MMVs. Based on the atomic $\ell_0$ norm approach, we theoretically show that the MMVs help improve the frequency recovery performance in terms of the number of recoverable frequencies. But unfortunately, the resulting optimization problem is NP-hard to solve. We therefore turn to the atomic norm approach and show that this convex relaxation approach can be cast as SDP and solved in a polynomial time. We also provide theoretical results showing that the atomic norm approach can successfully recover the frequencies under similar technical conditions as in \cite{candes2013towards,tang2012compressed}.

At a first glance, both the methods can be viewed as covariance-based by exploiting the structures of the data covariance matrix (obtained as if certain statistical assumptions for the source signals were satisfied). In particular, in both \eqref{formu:AL0_rankmin1} and \eqref{formu:AN_sdp}, the PSD Toeplitz matrix $T\sbra{\m{u}}$, which can be written as in \eqref{formu:VD}, can be viewed as the covariance matrix of the full data candidate $\m{Y}$ that is consistent with the observed data $\m{Y}_{\m{\Omega}}^o$ (see more details in the proofs of Theorems \ref{thm:AL0_rankmin} and \ref{thm:AN_SDP} in Section \ref{sec:proof}). The Toeplitz structure is explicitly given, the PSD is imposed by the first constraint, and the low rank is exploited in the objective function. The essential difference between the two methods lies in the way to exploit the low rank. To be specific, the atomic $\ell_0$ norm method utilizes this structure to the greatest extent possible by directly minimizing the rank, leading to a nonconvex optimization problem. In contrast, the atomic norm method uses convex relaxation and minimizes the nuclear norm or the trace norm of the matrix (note that the additional term $\tr\sbra{\m{W}}$ in \eqref{formu:AN_sdp} helps control the magnitude of $\m{u}$ and avoids a trivial solution). As a result, the theoretical guarantees that we provide actually state that the full data covariance matrix can be exactly recovered using the proposed methods given full or compressive data under certain conditions. Finally, note that source correlations in $\mbra{s_{kt}}$, if present, will be removed in the covariance estimate $T\sbra{\m{u}}$ in both \eqref{formu:AL0_rankmin1} and \eqref{formu:AN_sdp}, whereas they will be retained in the sample covariance used in conventional subspace methods.

The theoretical results presented above extend several existing results from the SMV to the MMV case or from the discrete to the continuous setting. To be specific, Theorem \ref{thm:AL0_guanrantee} is a continuous counterpart of \cite[Theorem 2.4]{chen2006theoretical} which deals with the conventional discrete setting. Theorem \ref{thm:AL0_guanrantee} shows that the number of recoverable frequencies can be increased in general as we take MMVs. This is practically relevant in array processing applications. But in a trivial case where all the sources are coherent, i.e., all the rows of $\mbra{s_{kt}}$ (and thus all the columns of $\m{Y}_{\m{\Omega}}^o$) are identical up to scaling factors, it holds that $\rank\sbra{\m{Y}_{\m{\Omega}}^o}=1$ as in the SMV case and hence, as expected, MMVs do not help improve the performance. Note also that it is generally difficult to compute $\spark\sbra{\cA_{\m{\Omega}}^1}$, except in the full data case where we have $\spark\sbra{\cA_{\m{\Omega}}^1}=N+1$ by the fact that any $N$ atoms in $\cA_{\m{\Omega}}^1$ are linear independent. An interesting topic in future studies will be the selection of $\m{\Omega}$, which in array processing corresponds to geometry design of the sensor array, such that $\spark\sbra{\cA_{\m{\Omega}}^1}$ is maximized.

Theorem \ref{thm:completedata} generalizes \cite[Theorem 1.2]{candes2013towards} from the SMV to the MMV case. Since Theorem \ref{thm:completedata} applies to all kinds of source signals, including the aforementioned trivial case, one cannot expect that the theoretical guarantee improves in the MMV case.

Theorem \ref{thm:incompletedata} generalizes \cite[Theorem I.1]{tang2012compressed} from the SMV to the MMV case. Note that in (\ref{formu:AN_bound}) the dependence of $M$ on $L$ is for controlling the probability of successful recovery. To make it clear, we consider the case when we seek to recover the columns of $\m{Y}^o$ independently via the SMV method in \cite{tang2012compressed}. When $M$ satisfies (\ref{formu:AN_bound}) with $L=1$, each column of $\m{Y}^o$ can be recovered with probability $1-\delta$. It follows that $\m{Y}^o$ can be exactly recovered with probability at least $1-L\delta$. In contrast, if we recover $\m{Y}^o$ via a single convex optimization problem that we propose, then with the same number of measurements the success probability is improved to $1-\sqrt{L}\delta$ (to see this, replace $\delta$ in (\ref{formu:AN_bound}) by $\sqrt{L}\delta$).

We note that in Theorem \ref{thm:incompletedata} the assumption on the phases $\m{\phi}_j$ is relaxed as compared to that in \cite[Theorem I.1]{tang2012compressed} (note that $\m{\phi}_j$'s are assumed in the latter drawn i.i.d. from a uniform distribution). This relaxation is significant in array processing since each $\m{\phi}_j$ corresponds to one source and therefore they do not necessarily obey an identical distribution. Note also that this assumption is weak in the sense that the sources can be coherent, resulting in the aforementioned trivial case. To see this, suppose that the rows of $\mbra{s_{kt}}$ are i.i.d. Gaussian with zero mean and covariance of rank one. Then the sources are certain to be independent and coherent. This explains why the theoretical guarantee given in Theorem \ref{thm:incompletedata} does not improve in the presence of MMVs. In this sense, therefore, the results of Theorems \ref{thm:completedata} and \ref{thm:incompletedata} are referred to as \emph{worst case} analysis.

Our contribution by Theorems \ref{thm:completedata} and \ref{thm:incompletedata} is showing that in the presence of MMVs we can confidently recover the frequencies via a single convex optimization problem by exploiting the joint sparsity therein. Although the worst case analysis we provide cannot shed light on the advantage of MMVs, numerical simulations provided in Section \ref{sec:simulation} indeed imply that the proposed atomic norm approach significantly improves the recovery performance when the source signals are at general positions. We pose such \emph{average case} analysis as a future work.

\section{Connections to Prior Art} \label{sec:connection}

\subsection{Grid-based Joint Sparse Recovery}
The JSFR problem concerned in this paper has been widely studied within the CS framework, typically under the topic of DOA estimation. It has been popular in the past decade to assume that the true frequencies lie on a fixed grid since, according to conventional wisdom on CS, the signal needs to be sparsely represented under a finite discrete dictionary. Now recall the atomic $\ell_p$ norm in (\ref{formu:AL0}) and (\ref{formu:atomicnorm}) with $p=0$ and $1$, respectively, that can be written collectively as:
\equ{\norm{\m{Y}}_{\cA,p}=\inf\lbra{\sum_k \twon{\m{s}_k}^p: \m{Y}=\sum_k\m{a}\sbra{f_k}\m{s}_k, f_k\in\bT}, \label{formu:atomp}}
where $\m{s}_k\in\bC^{1\times L}$. Consequently, the atomic $\ell_0$ norm and the atomic norm can be viewed, respectively, as the continuous counterparts of the $\ell_{2,0}$ norm and the $\ell_{2,1}$ norm in grid-based joint sparse recovery methods (see, e.g., \cite{malioutov2005sparse,hyder2010direction}). It is worth noting that for the existing grid-based methods one cannot expect exact frequency recovery since in practice the true frequencies typically do not lie on the grid. Moreover, even if this on-grid assumption is satisfied, the existing coherence or RIP-based analysis in the discrete setting is very conservative, as compared to the results in this paper, due to high coherence in the case of a dense grid. Readers are referred to \cite{candes2013towards} for detailed discussions on the SMV case.

\subsection{Gridless Joint Sparse Recovery}
To the best of our knowledge, the only discretization-free/gridless technique for JSFR was introduced in \cite{yang2014discretization} prior to this work, termed as the sparse and parametric approach (SPA). Different from the atomic norm technique proposed in this paper, SPA is from a statistical perspective and based on a weighted covariance fitting criterion \cite{stoica2011spice}. But we show next that the two methods are strongly connected. Consider the full data case as an example. In the limiting noiseless case, SPA solves the following problem:
\equ{ \min_{\m{u}\in\bC^N, T\sbra{\m{u}}\geq\m{0}} \tr\sbra{\widehat{\m{R}}\mbra{T\sbra{\m{u}}}^{-1}\widehat{\m{R}}}+\tr\sbra{T\sbra{\m{u}}}, \label{formu:SPA}}
where $\widehat{\m{R}}=\frac{1}{L}\m{Y}^o{\m{Y}^o}^H$ denotes the sample covariance matrix. Let $\m{V}=\frac{1}{L}\sbra{{\m{Y}^o}^H\m{Y}^o}^{\frac{1}{2}}\in\bC^{L\times L}$. Then we have the following equalities/equivalences:
\equ{\begin{split} \text{(\ref{formu:SPA})}
&= \min_{\m{u}, T\sbra{\m{u}}\geq\m{0}} \tr\sbra{\sbra{\m{Y}^o\m{V}}^H \mbra{T\sbra{\m{u}}}^{-1} \m{Y}^o\m{V}} + \tr\sbra{T\sbra{\m{u}}}\\
&= \min_{\m{W},\m{u}} \tr\sbra{\m{W}} + \tr\sbra{T\sbra{\m{u}}},\\
&\quad \st \begin{bmatrix} \m{W} & \sbra{\m{Y}^o\m{V}}^H \\  \m{Y}^o\m{V} & T\sbra{\m{u}} \end{bmatrix}\geq\m{0}\\
&= 2\sqrt{N}\atomn{\m{Y}^o\m{V}},
\end{split} \notag}
where the last equality follows from Theorem \ref{thm:AN_SDP}.
This means that SPA actually computes the atomic norm of
\equ{\m{Y}^o\m{V}=\sum_{k=1}^K \m{a}\sbra{f_k}\sbra{\m{s}_{k}\m{V}}. \label{formu:observmodel_SPA}}
Therefore, SPA can be interpreted as an atomic norm approach with modification of the source signals.
In the SMV case where $\m{V}$ is a positive scalar, the two techniques are exactly equivalent, which has been shown in \cite{yang2015gridless}. While details are omitted, note that a similar result holds in the compressive data case.

\section{Proofs} \label{sec:proof}

The proofs of Theorems 1-5 are provided in this section. While our proofs generalize several results in the literature either from the SMV to the MMV case or from the discrete to the continuous setting, note that they are not straightforward. For example, the proof of Theorem \ref{thm:AN_SDP} does not follow from \cite{tang2012compressed} in the SMV case but is motivated by \cite{yang2014discretization,yang2015gridless}. The main challenge of the proofs of Theorems \ref{thm:completedata} and \ref{thm:incompletedata} lie in how to construct and deal with vector-valued dual polynomials instead of the scalar-valued ones in \cite{candes2013towards} and \cite{tang2012compressed}. Moreover, the proof of Theorem \ref{thm:completedata} forms the basis of the proof of Theorem \ref{thm:incompletedata}. Some inaccuracy in \cite{tang2012compressed} is also pointed out and corrected.

\subsection{Proof of Theorem \ref{thm:AL0_rankmin}} \label{sec:proofAL0}
Let $K=\norm{\m{Y}}_{\cA,0}$ and $K^*=\rank\sbra{T\sbra{\m{u}^*}}$, where $\m{u}^*$ denotes an optimal solution of $\m{u}$ in (\ref{formu:AL0_rankmin}). It suffices to show that $K=K^*$. On one hand, using the Vandermonde decomposition, we have that $T\sbra{\m{u}^*}=\sum_{j=1}^{K^*} p_j\m{a}\sbra{f_j}\m{a}^H\sbra{f_j}$. Moreover, the fact that $\m{Y}$ lies in the range space of $T\sbra{\m{u}^*}$ implies that there exist $\m{s}_j\in\bC^{1\times L}$, $j\in\mbra{K^*}$ such that $\m{Y}=\sum_{j=1}^{K^*}\m{a}\sbra{f_j}\m{s}_{j}$. It follows from the definition of $\norm{\m{Y}}_{\cA,0}$ that $K\leq K^*$.

On the other hand, let $\m{Y}=\sum_{j=1}^K \m{a}\sbra{f_j}\m{s}_j$ be an atomic decomposition of $\m{Y}$. Let $T\sbra{\m{u}}=\sum_{j=1}^K p_j\m{a}\sbra{f_j}\m{a}^H\sbra{f_j}$ and $\m{W}=\sum_{j=1}^K p_j^{-1}\m{s}_j^H\m{s}_j$ for arbitrary $p_j>0$, $j\in\mbra{K}$. Then,
\equ{\begin{bmatrix}\m{W} & \m{Y}^H \\ \m{Y} & T\sbra{\m{u}}\end{bmatrix}= \sum_{j=1}^K p_j \begin{bmatrix}p_j^{-1}\m{s}_j^H\\ \m{a}\sbra{f_j}\end{bmatrix} \begin{bmatrix}p_j^{-1}\m{s}_j & \m{a}\sbra{f_j}^H\end{bmatrix}\geq\m{0}. \notag}
This means that $\sbra{\m{W},\m{u}}$ defines a feasible solution of (\ref{formu:AL0_rankmin}). Consequently, $K^*\leq \rank\sbra{T\sbra{\m{u}}}=K$.

\subsection{Proof of Theorem \ref{thm:AN_SDP}}
We use the following identity whenever $\m{R}\geq\m{0}$:
\equ{\m{y}^H\m{R}^{-1}\m{y}=\min t, \st \begin{bmatrix}t & \m{y}^H \\ \m{y} & \m{R}\end{bmatrix}\geq0. \label{formu:yRy}}
In fact, (\ref{formu:yRy}) is equivalent to defining $\m{y}^H\m{R}^{-1}\m{y}\coloneqq\lim_{\sigma\rightarrow0_+} \m{y}^H\sbra{\m{R}+\sigma\m{I}}^{-1}\m{y}$ when $\m{R}$ loses rank. We also use the following lemma.

\begin{lem}[\cite{yang2015gridless}] Given $\m{R}=\m{A}\m{A}^H\geq\m{0}$, it holds that $\m{y}^H\m{R}^{-1}\m{y}=\min\twon{\m{s}}^2, \st \m{A}\m{s}=\m{y}$. \label{lem:lem1}
\end{lem}

Now we prove Theorem \ref{thm:AN_SDP}. It follows from the constraint in \eqref{formu:AN_SDP} that $T\sbra{\m{u}}\geq\m{0}$ and $\m{W}\geq \m{Y}^H\mbra{T\sbra{\m{u}}}^{-1} \m{Y}$. So, it suffices to show that
\equ{\begin{split}\norm{\m{Y}}_{\cA}=
&\min_{\m{u}} \frac{\sqrt{N}}{2}u_1 + \frac{1}{2\sqrt{N}}\tr\sbra{\m{Y}^H\mbra{T\sbra{\m{u}}}^{-1} \m{Y}},\\
&\st T\sbra{\m{u}}\geq\m{0}, \end{split} \label{formu:AN_sdp222}}
where $u_1$ is the first entry of $\m{u}$.
Let $T\sbra{\m{u}}=\m{A}\m{P}\m{A}^H=\mbra{\m{A}\m{P}^{\frac{1}{2}}} \mbra{\m{A}\m{P}^{\frac{1}{2}}}^H$ be any feasible Vandermonde decomposition, where $\m{A}=\m{A}\sbra{\m{f}}=\mbra{\dots,\m{a}\sbra{f_j},\dots}$ and $\m{P}=\diag\sbra{\dots,p_j,\dots}$ with $p_j>0$. It follows that $u_1=\sum p_j$. For the $t$th column of $\m{Y}$, say $\m{y}_{: t}$, it holds by Lemma \ref{lem:lem1} that
\equ{\begin{split}\m{y}_{: t}^H\mbra{T\sbra{\m{u}}}^{-1}\m{y}_{: t}
&= \min_{\m{v}} \twon{\m{v}}^2, \st \m{A}\m{P}^{\frac{1}{2}}\m{v}=\m{y}_{: t}\\
&= \min_{\m{s}} \twon{\m{P}^{-\frac{1}{2}}\m{s}}^2, \st \m{A}\m{s}=\m{y}_{: t}\\
&= \min_{\m{s}} \m{s}^H\m{P}^{-1}\m{s}, \st \m{A}\m{s}=\m{y}_{: t}.
\end{split} \notag}
It follows that
\equ{\begin{split}
\tr\sbra{\m{Y}^H\mbra{T\sbra{\m{u}}}^{-1}\m{Y}}
&= \sum_{t=1}^N \m{y}_{: t}^H\mbra{T\sbra{\m{u}}}^{-1}\m{y}_{: t} \\
&= \min_{\m{S},\m{A}\sbra{\m{f}}\m{S}=\m{Y}} \tr\sbra{\m{S}^H\m{P}^{-1}\m{S}}.\\
\end{split} \notag}
We complete the proof via the following equalities:
\equ{\begin{split}
& \min_{\m{u}} \frac{\sqrt{N}}{2}u_1 + \frac{1}{2\sqrt{N}}\tr\sbra{\m{Y}^H\mbra{T\sbra{\m{u}}}^{-1} \m{Y}}\\
=& \min_{\substack {\m{f},\m{p}\succeq\m{0}, \m{S}\\ \m{A}\sbra{\m{f}}\m{S}=\m{Y}}}\frac{\sqrt{N}}{2}\sum_j p_j + \frac{1}{2\sqrt{N}}\tr\sbra{\m{S}^H\m{P}^{-1}\m{S}}\\
=& \min_{\substack {\m{f},\m{p}\succeq\m{0}, \m{S}\\ \m{A}\sbra{\m{f}}\m{S}=\m{Y}}} \frac{\sqrt{N}}{2}\sum_j p_j + \frac{1}{2\sqrt{N}}\sum_j \twon{\m{S}_j}^2p_j^{-1}\\
=& \min_{\m{f},\m{S}} \sum_j\twon{\m{S}_j}, \st \m{Y}=\m{A}\sbra{\m{f}}\m{S} \\
=& \norm{\m{Y}}_{\cA},
\end{split} \label{formu:AN_equa}}
where the optimal solution of $p_j$ equals $\frac{1}{\sqrt{N}}\twon{\m{S}_j}$ and the last equality follows from \eqref{formu:atomp}.

\subsection{Proof of Theorem \ref{thm:AL0_guanrantee}}
We use contradiction. Suppose that there exists $\widetilde{\m{Y}}\neq \m{Y}^o$ satisfying that $\widetilde{\m{Y}}_{\m{\Omega}}= \m{Y}^o_{\m{\Omega}}$ and $\widetilde{K}\coloneqq\norm{\widetilde{\m{Y}}}_{\cA,0}\leq \norm{\m{Y}^o}_{\cA,0}= K$. Let $\widetilde{\m{Y}}= \sum_{k=1}^{\widetilde K} \m{a}\sbra{\widetilde f_j}\widetilde{\m{s}}_j$ be an atomic decomposition. Also let $\m{A}_1=\mbra{\m{a}\sbra{f}}_{f\in\cT\backslash{\lbra{\widetilde f_j}}}$ (the matrix consisting of those $\m{a}\sbra{f}$, $f\in\cT\backslash{\lbra{\widetilde f_j}}$), $\m{A}_{12}=\mbra{\m{a}\sbra{f}}_{f\in\cT\cap{\lbra{\widetilde f_j}}}$ and $\m{A}_2=\mbra{\m{a}\sbra{f}}_{f\in{\lbra{\widetilde f_j}}\backslash\cT}$. In addition, let $K_{12}=\abs{\cT\cap{\lbra{\widetilde f_j}}}$ and $\m{A}=\begin{bmatrix}\m{A}_{1} & \m{A}_{12} & \m{A}_2\end{bmatrix}$. Then we have
$\m{Y}^o= \begin{bmatrix}\m{A}_1 & \m{A}_{12}\end{bmatrix} \begin{bmatrix}\m{S}_{1} \\ \m{S}_{12}\end{bmatrix}$ and $\widetilde{\m{Y}}= \begin{bmatrix}\m{A}_{12} & \m{A}_{2}\end{bmatrix} \begin{bmatrix}\m{S}_{21} \\ \m{S}_{2}\end{bmatrix}$, where $\m{S}_1$, $\m{S}_{12}$, $\m{S}_{21}$ and $\m{S}_{2}$ are properly defined. It follows that $\m{Y}^o-\widetilde{\m{Y}}= \m{A} \m{\Upsilon}\neq\m{0}$, where $\m{\Upsilon}=\begin{bmatrix}\m{S}_{1} \\ \m{S}_{12}-\m{S}_{21}\\ -\m{S}_{2}\end{bmatrix}\neq\m{0}$.
On the other hand, it follows from $\widetilde{\m{Y}}_{\m{\Omega}}= \m{Y}^o_{\m{\Omega}}$ that $\m{A}_{\m{\Omega}} \m{\Upsilon}=\m{0}$. Note that $\m{A}_{\m{\Omega}}$ is composed of atoms in $\cA_{\m{\Omega}}^1$ and has a nontrivial null space since we have shown that $\m{\Upsilon}\neq\m{0}$. Then,
\equ{\rank\sbra{\m{A}_{\m{\Omega}}}\geq \spark\sbra{\cA_{\m{\Omega}}^1}-1. \label{formu:rankbound}}
Moreover, for the nullity (dimension of the null space) of $\m{A}_{\m{\Omega}}$ it holds that
\equ{\begin{split}
\text{nullity}\sbra{\m{A}_{\m{\Omega}}}
&\geq \rank\sbra{\m{\Upsilon}}\\
&\geq \rank\sbra{\begin{bmatrix}\m{S}_{1} \\ \m{S}_{12}\end{bmatrix}} - \rank\sbra{\begin{bmatrix}\m{0} \\ \m{S}_{21}\end{bmatrix}}\\
&\geq \rank\sbra{\m{Y}_{\m{\Omega}}^o} - K_{12}.\end{split} \label{formu:nullitybound}}
Consequently, the equality
\equ{\#\text{columns of }\m{A}_{\m{\Omega}}=\rank \sbra{\m{A}_{\m{\Omega}}}+ \text{nullity}\sbra{\m{A}_{\m{\Omega}}} \notag}
together with (\ref{formu:rankbound}) and (\ref{formu:nullitybound}) yields that
$K+\widetilde K-K_{12}\geq \spark\sbra{\cA_{\m{\Omega}}^1}-1+ \rank\sbra{\m{Y}_{\m{\Omega}}^o} - K_{12}$. Therefore,
\equ{2K\geq K+\widetilde K\geq \spark\sbra{\cA_{\m{\Omega}}^1}-1+\rank \sbra{\m{Y}_{\m{\Omega}}^o}, \notag}
which contradicts the condition in (\ref{Kbound}).

To show the uniqueness part, note that the condition in (\ref{Kbound}) implies that $K< \spark\sbra{\cA_{\m{\Omega}}^1}-1$ since $\rank \sbra{\m{Y}_{\m{\Omega}}^o}\leq K$. According to the definition of spark, any $K$ atoms in $\cA_{\m{\Omega}}^1$ are linearly independent. Therefore, the atomic decomposition is unique given the set of frequencies $\cT=\lbra{f_j}_{j=1}^K$. Now suppose there exists another decomposition $\m{Y}^o=\sum_{j=1}^{\widetilde K} \m{a}\sbra{\widetilde f_j}\widetilde{\m{s}}_j$, where $\widetilde K\leq K$ and $\widetilde f_{j_0}\notin\cT$ for some $j_0\in\mbra{\widetilde{K}}$. Note that we have used the same notations for simplicity and we similarly define the other notations. Once again we have that $\m{\Upsilon}\neq \m{0}$ since $\m{A}_2$ is nonempty and $\m{S}_2\neq\m{0}$. The rest of the proof follows from the same arguments as above.

\subsection{Proof of Theorem \ref{thm:completedata}}
The proof of Theorem \ref{thm:completedata} generalizes that in \cite{candes2013towards} (and re-organized in \cite{tang2012compressed}) from the SMV to the MMV case. The main challenge is how to construct and deal with a vector-valued dual polynomial induced by the MMV problem, instead of the scalar-valued one in \cite{candes2013towards}. Since our proof follows similar procedures as in \cite{candes2013towards} and because of the page limit, we only highlight the key steps. Readers are referred to Section 5 of the technical report \cite{yang2014exact1} for the detailed proof.

Following from \cite{tang2012compressed}, we can consider an equivalent case of symmetric data index set $\m{J}=\lbra{-2n,\dots,2n}$, where $n=\left\lfloor\frac{N-1}{4}\right\rfloor$, instead of the set specified in \eqref{formu:observmodel1}. As in \cite{candes2013towards}, we link Theorem \ref{thm:completedata} to a dual polynomial. In particular, Theorem \ref{thm:completedata} holds if there exists a vector-valued dual polynomial $Q: \bT\rightarrow \bC^{1\times L}$,
\equ{Q(f)=\m{a}(f)^H\m{V} \label{formu:dualpoly1}}
satisfying that
\lentwo{\equa{ Q\sbra{f_k}
&=& \m{\phi}_k, \quad f_k\in\cT, \label{formu:cons1}\\ \twon{Q\sbra{f}}
&<& 1, \quad f\in\bT\backslash\cT, \label{formu:cons2}
}}where the coefficient matrix $\m{V}\in\bC^{\abs{\m{J}}\times L}$. The following proof is devoted to construction of $Q(f)$ under the assumptions of Theorem \ref{thm:completedata}.

Inspired by \cite{candes2013towards}, we let
\equ{Q\sbra{f}=\sum_{f_j\in\cT}\m{\alpha}_j \cK\sbra{f-f_j} + \sum_{f_j\in\cT}\m{\beta}_j \cK'\sbra{f-f_j}, \label{formu:dualpoly}}
where $\cK\sbra{f}$ is the squared Fej\'{e}r kernel
\equ{\cK\sbra{f}=\mbra{\frac{\sin(\pi (n+1)f)}{(n+1)\sin\sbra{\pi f}}}^4=\sum_{j=-2n}^{2n}g_je^{-i2\pi jf} \label{formu:kernel}}
in which $g_j$ are constant, $\cK'$ denotes the first-order derivative of $\cK$, and the coefficients $\m{\alpha}_j,\m{\beta}_j\in\bC^{1\times L}$ are specified by imposing \eqref{formu:cons1} and
\equ{Q'\sbra{f_k}=\m{0}, \quad f_k\in\cT.\label{formu:dualpolycons2}}
The equations in \eqref{formu:cons1} and \eqref{formu:dualpolycons2} can be combined into the linear system of equations:
\equ{ \begin{bmatrix}\m{D}_0 & c_0^{-1}\m{D}_1 \\ -c_0^{-1}\m{D}_1 & -c_0^{-2}\m{D}_2\end{bmatrix} \begin{bmatrix}\m{\alpha}\\ c_0\m{\beta}\end{bmatrix} = \begin{bmatrix}\m{\Phi}\\ \m{0}\end{bmatrix}, \label{formu:linsysalpbet}}
where the coefficient matrix $\m{D}\coloneqq\begin{bmatrix}\m{D}_0 & c_0^{-1}\m{D}_1 \\ -c_0^{-1}\m{D}_1 & -c_0^{-2}\m{D}_2\end{bmatrix}$ only depends on the frequency set $\cT$, $c_0=\sqrt{\cK''(0)}$ is a constant, $\m{\Phi}=\begin{bmatrix}\m{\phi}_1^T,\dots,\m{\phi}_K^T\end{bmatrix}^T\in\bC^{K\times L}$, $\m{\alpha}=\begin{bmatrix}\m{\alpha}_1^T,\dots,\m{\alpha}_K^T\end{bmatrix}^T\in\bC^{K\times L}$ and $\m{\beta}\in\bC^{K\times L}$ is similarly defined. Using the fact that the coefficient matrix in \eqref{formu:linsysalpbet} is close to identity \cite{candes2013towards}, we next prove that $\begin{bmatrix}\m{\alpha}\\ c_0\m{\beta}\end{bmatrix}$ is close to $\begin{bmatrix}\m{\Phi}\\ \m{0}\end{bmatrix}$. Different from the SMV case in which $\m{\alpha}_j$ and $\m{\beta}_j$ are scalars, the difficulty in our proof is how to quantify this closeness. To do this, we define the $\ell_{2,\infty}$ matrix norm and its induced operator norm as follows.

\begin{defi} We define the $\ell_{2,\infty}$ norm of $\m{X}\in\bC^{d_1\times d_2}$ as
\equ{\twoinfn{\m{X}}=\max_{j} \twon{\m{X}_j} \notag}
and its induced norm of a linear operator $\m{\cP}: \bC^{d_1\times d_2}\rightarrow\bC^{d_3\times d_2}$ as
\equ{\twoinfn{\m{\cP}}=\sup_{\m{X}\neq\m{0}} \frac{\twoinfn{\m{\cP}\m{X}}}{\twoinfn{\m{X}}}= \sup_{\twoinfn{\m{X}}\leq1} \twoinfn{\m{\cP}\m{X}}, \notag}
where $\m{X}_j$ denotes the $j$th row of $\m{X}$, and $d_1$, $d_2$ and $d_3$ are positive integers. \label{def:matrixnorm}
\end{defi}

By Definition \ref{def:matrixnorm}, we have that $\twoinfn{\m{\Phi}}=1$ and expect to bound $\twoinfn{\m{\alpha}}$ and $\twoinfn{\m{\beta}}$ using the induced norm of the operators $\m{D}_j$, $j=0,1,2$. To do so, we calculate the induced norm first. Interestingly, the induced $\ell_{2,\infty}$ norm is identical to the $\ell_{\infty}$ norm, which is stated in the following result.

\begin{lem}[\cite{yang2014exact1}] $\twoinfn{\m{\cP}}=\inftyn{\m{\cP}}$ for any linear operator $\m{\cP}$ defined by a matrix $\m{P}$ such that $\m{\cP}\m{X}=\m{P}\m{X}$ for any $\m{X}$ of proper dimension. \label{lem:twoinfinitynorm}
\end{lem}

By Lemma \ref{lem:twoinfinitynorm} the $\ell_{2,\infty}$ operator norm of $\m{D}_j$, $j=0,1,2$ equals their $\ell_{\infty}$ norm that has been derived in \cite{candes2013towards}. Then, under the assumptions of Theorem \ref{thm:completedata} and using the results in \cite{candes2013towards}, we can show that
\lentwo{\equa{ \twoinfn{\m{\alpha}-\m{\Phi}}
&\leq& 8.824\times10^{-3}, \label{formu:alphabound}\\ \twoinfn{\m{\beta}}
&\leq& \frac{1.647}{n}\times10^{-2}. \label{formu:betabound}}}

Finally, we complete the proof by showing that the constructed polynomial $Q\sbra{f}$ satisfies (\ref{formu:cons2}) using \eqref{formu:alphabound}, \eqref{formu:betabound} and the bounds on $\cK\sbra{f-f_k}$ and its derivatives given in \cite{candes2013towards}. As in \cite{candes2013towards}, we divide $\bT$ into several intervals that are either neighborhood of or far from some $f_k\in\cT$. If $f$ is far from every $f_k\in\cT$, then we can show that $\twon{Q\sbra{f}}\leq 0.99992$. Otherwise, we can show that on the neighborhood of $f_k\in\cT$, the second derivative of $\twon{Q\sbra{f}}^2$ is negative. This means that $\twon{Q\sbra{f}}^2$ is a strictly concave function and achieves its maximum $1$ at the only stationary point $f_k$ by \eqref{formu:dualpolycons2}. So we can conclude \eqref{formu:cons2} and complete the proof.

\subsection{Proof of Theorem \ref{thm:incompletedata}}
The proof of Theorem \ref{thm:completedata} in the last subsection forms the basis of the proof of Theorem \ref{thm:incompletedata} that will be given following similar steps as in \cite{tang2012compressed}. As in the full data case, we only highlight the key steps of our proof and interested readers are referred to \cite[Section 6]{yang2014exact1} for the details. Similarly, we can also consider the symmetric case of $\m{J}=\lbra{-2n,\dots,2n}$ and start with the dual certificate. In particular, $\m{Y}^o=\sum_{k=1}^K c_k\m{a}\sbra{f_k,\m{\phi}_k}$ is the unique optimizer to (\ref{formu:ANmin}) and provides the unique atomic decomposition satisfying that $\atomn{\m{Y}^o}=\sum_{k=1}^K c_k$ if 1) $\lbra{\m{a}_{\m{\Omega}}\sbra{f_k}}_{f_k\in\cT}\subset \cA_{\m{\Omega}}^1$ are linearly independent and 2) there exists a vector-valued dual polynomial $\overline{Q}(f)=\m{a}^H(f)\m{V}\in \bC^{1\times L}$ as in \eqref{formu:dualpoly1} satisfying \eqref{formu:cons1}, \eqref{formu:cons2} and the additional constraint that
\equ{\m{V}_j= \m{0}, \quad j\notin \m{\Omega}. \label{formu:cons3_inc}}
Note that the condition of linear independence above is necessary to prove the uniqueness part but is neglected in \cite{tang2012compressed}. We will show later that this condition is satisfied for free when we construct the dual polynomial $\overline{Q}\sbra{f}$ under the assumptions of Theorem \ref{thm:incompletedata}. As in \cite{tang2012compressed}, we consider an equivalent Bernoulli observation model in which the samples indexed by $\m{J}$ are observed independently with probability $p=\frac{M}{4n}$. In mathematics, let $\lbra{\delta_j}_{j\in\m{J}}$ be i.i.d. Bernoulli random variables such that
\equ{\bP\sbra{\delta_j=1}=p,}
where $\delta_j=1$ or $0$ indicates whether the $j$th entry in $\m{J}$ is observed or not. It follows that the sampling index set $\m{\Omega}=\lbra{j:\delta_j=1}$.

Inspired by \cite{tang2012compressed}, we let
\equ{\overline Q\sbra{f}=\sum_{f_j\in\cT}\m{\alpha}_j \overline\cK\sbra{f-f_j} + \sum_{f_j\in\cT}\m{\beta}_j \overline\cK'\sbra{f-f_j}, \label{formu:randpoly}}
where $\overline\cK\sbra{f}$ is a random analog of $\cK\sbra{f}$ as defined in \eqref{formu:kernel}:
\equ{\overline\cK\sbra{f}= \sum_{j=-2n}^{2n} \delta_jg_n\sbra{j}e^{-i2\pi jf}. \label{formu:cKbar}}
It is clear that $\bE \overline\cK\sbra{f}=p\cK\sbra{f}$ and similar result holds for its derivatives. Again, we impose for the coefficients $\m{\alpha}_j$, $\m{\beta}_j\in\bC^{1\times L}$ that
\equ{\overline{\m{D}} \begin{bmatrix}\m{\alpha}\\ c_0\m{\beta}\end{bmatrix} = \begin{bmatrix}\m{\Phi}\\ \m{0}\end{bmatrix},}
where $\overline{\m{D}}$ is a random analog of $\m{D}$ in \eqref{formu:linsysalpbet} with $\bE \overline{\m{D}}=p{\m{D}}$. It is clear that $\overline{Q}(f)$ above already satisfies (\ref{formu:cons1}) and \eqref{formu:cons3_inc}. The remaining task is showing that it also satisfies (\ref{formu:cons2}) under the assumptions of Theorem \ref{thm:incompletedata}.

Let $Q(f)$ be the dual polynomial in \eqref{formu:dualpoly1} that is the full data case counterpart of $\overline{Q}(f)$. As in \cite{tang2012compressed}, we need to show that $\overline{Q}(f)$ (and its derivatives) is tightly concentrated around $Q(f)$ (and its derivatives) when the sample size $M$ satisfies \eqref{formu:AN_bound}. To do this, define two events
\lentwo{\equa{\cE_{1}
&=&\lbra{\twon{p^{-1}\overline{\m{D}}-\m{D}}\leq\frac{1}{4}},\\ \cE_2
&=&\lbra{\sup_{f\in\bT_{\text{grid}}} c_0^{-l} \twon{\overline Q^{\sbra{l}}-Q^{\sbra{l}}} \leq \frac{\epsilon}{3}, l=0,1,2,3}}
}where $\bT_{\text{grid}}\subset\bT$ and $\epsilon>0$ are a set of discrete points and a small number, respectively, to specify. It has been shown in \cite{tang2012compressed} that $\overline{\m{D}}$ is invertible on $\cE_{1}$ which happens with probability at least $1-\delta$ if
\equ{M\geq C_1 K\log \frac{K}{\delta}}
and if the frequency separation condition is satisfied, where $C_1$ is constant. Note that the aforementioned linear independence of $\lbra{\m{a}_{\m{\Omega}}\sbra{f_k}}_{f_k\in\cT}\subset \cA_{\m{\Omega}}^1$ can be shown based on this result (see \cite[Lemma 6.4]{yang2014exact1}). We next focus on the case when $\cE_{1}$ happens. It follows that
\equ{ \begin{bmatrix}\m{\alpha}\\ c_0\m{\beta}\end{bmatrix} = \overline{\m{D}}^{-1}\begin{bmatrix}\m{\Phi}\\ \m{0}\end{bmatrix}= \overline{\m{L}}\m{\Phi},}
where $\overline{\m{L}}\in\bC^{2K\times K}$ denotes the left part of $\overline{\m{D}}^{-1}$. Therefore, as in \cite{tang2012compressed}, we have that
\equ{c_0^{-l} \mbra{\overline Q^{\sbra{l}}(f)-Q^{\sbra{l}}(f)} = H_1(f)\m{\Phi} + H_2(f)\m{\Phi},}
where $H_1(f),H_2(f)\in\bC^{1\times K}$ are as defined and bounded in \cite{tang2012compressed}. The main difference from \cite{tang2012compressed} lies in the fact that $\m{\Phi}$ is a $K\times L$ matrix instead of a $K\times 1$ vector. To show that both $\twon{H_1(f)\m{\Phi}}$ and $\twon{H_2(f)\m{\Phi}}$ are concentrated around 0 with high probability, we need the following vector-form Hoeffding's inequality that can be proven based on \cite[Theorem 1.3]{tropp2012user}.

\begin{lem}[\cite{yang2014exact1}] Let the rows of $\m{\Phi}\in\bC^{K\times L}$ be sampled independently on the complex hyper-sphere $\bS^{2L-1}$ with zero mean. Then, for all $\m{w}\in\bC^K$, $\m{w}\neq\m{0}$, and $t\geq0$,
\equ{\bP\sbra{\twon{\m{w}^H\m{\Phi}}\geq t } \leq \sbra{L+1}e^{-\frac{t^2}{8\twon{\m{w}}^2}}. \notag} \label{lem:vectorHoeffding}
\end{lem}

Using Lemma \ref{lem:vectorHoeffding} we can show that $\cE_2$ happens with probability at least $1-\delta$ if
\equ{\begin{split} M\geq C_2\frac{1}{\epsilon^2}\max\Big\{
&\log\frac{\abs{\bT_{\text{grid}}}}{\delta} \log\frac{L\abs{\bT_{\text{grid}}}}{\delta}, \\
& \quad K\log\frac{K}{\delta} \log\frac{L\abs{\bT_{\text{grid}}}}{\delta}\Big\} \end{split} \label{formu:Mboundgrid}}
among other assumptions in Theorem \ref{thm:incompletedata}, where $C_2$ is constant.
This result is then extended, as in \cite{tang2012compressed}, from $\bT_{\text{grid}}$ to the whole unit circle $\bT$ by choosing some $\bT_{\text{grid}}$ satisfying that
\equ{\abs{\bT_{\text{grid}}}< \frac{3C_3\sqrt{L}n^3}{\epsilon}, \label{eq:gridcard}}
where $C_3$ is also constant. This means that $\overline{Q}(f)$ (and its derivatives) is concentrated around $Q(f)$ (and its derivatives) with high probability. Now we are ready to complete the proof by showing that $\twon{\overline{Q}(f)}$ satisfies (\ref{formu:cons2}) using the properties of $Q(f)$ shown in the last subsection and by properly choosing $\epsilon$. In particular, letting $\epsilon=10^{-5}$, $\twon{\overline{Q}(f)}$ can still be well bounded by 1 from above when $f$ is far from every $f_k\in\cT$. When $f$ is in the neighborhood of some $f_k\in\cT$, the second derivative of $\twon{\overline{Q}\sbra{f}}^2$ is concentrated around the second derivative of $\twon{Q\sbra{f}}^2$ and thus it is negative. It follows that $\twon{\overline{Q}\sbra{f}}^2$ is strictly concave and achieves the maximum $1$ at the only stationary point $f_k$. Finally, to close the proof, note that inserting \eqref{eq:gridcard} into \eqref{formu:Mboundgrid} resulting in the bound in \eqref{formu:AN_bound}.

\section{Numerical Simulations} \label{sec:simulation}

\subsection{Full Data}
We consider the full data case and test the frequency recovery performance of the proposed atomic norm method with respect to the frequency separation condition. In particular, we consider two types of frequencies, equispaced and random, and two types of source signals, uncorrelated and coherent. We fix $N=128$ and vary $\Delta_{\text{min}}$ (a lower bound of the minimum separation of frequencies) from $1.05N^{-1}$ (or $0.9N^{-1}$ for random frequencies) to $2N^{-1}$ at a step of $0.05N^{-1}$. In the case of equispaced frequencies, for each $\Delta_{\text{min}}$ we generate a set of frequencies $\cT$ of the maximal cardinality $\lfloor\Delta_{\text{min}}^{-1}\rfloor$ with frequency separation $\Delta_{\cT}=\frac{1}{\lfloor\Delta_{\text{min}}^{-1}\rfloor}\geq\Delta_{\text{min}}$. In the case of random frequencies, we generate the frequency set $\cT$, $\Delta_{\cT}\geq\Delta_{\text{min}}$, by repetitively adding new frequencies (generated uniformly at random)  till no more can be added. Therefore, any two adjacent frequencies in $\cT$ are separated by a value in the interval $[\Delta_{\text{min}},2\Delta_{\text{min}})$. It follows that $\abs{\cT}\in\left(\frac{1}{2}\Delta_{\text{min}}^{-1},\Delta_{\text{min}}^{-1}\right]$. We empirically find that $\bE\abs{\cT}\approx\frac{3}{4}\Delta_{\text{min}}^{-1}$ which is the mid-point of the interval above.

\begin{figure}
\centering
  \subfigure[Equispaced frequencies]{
    \label{Fig:Complete_L_UniformFreq}
    \includegraphics[width=3in]{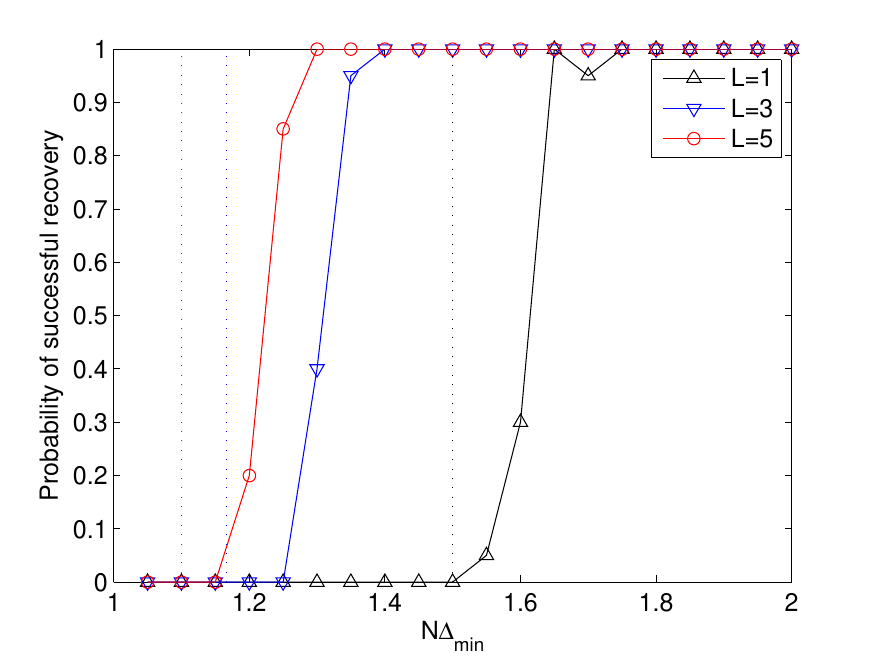}}
  \subfigure[Random frequencies]{
    \label{Fig:Complete_L_NonuniformFreq}
    \includegraphics[width=3in]{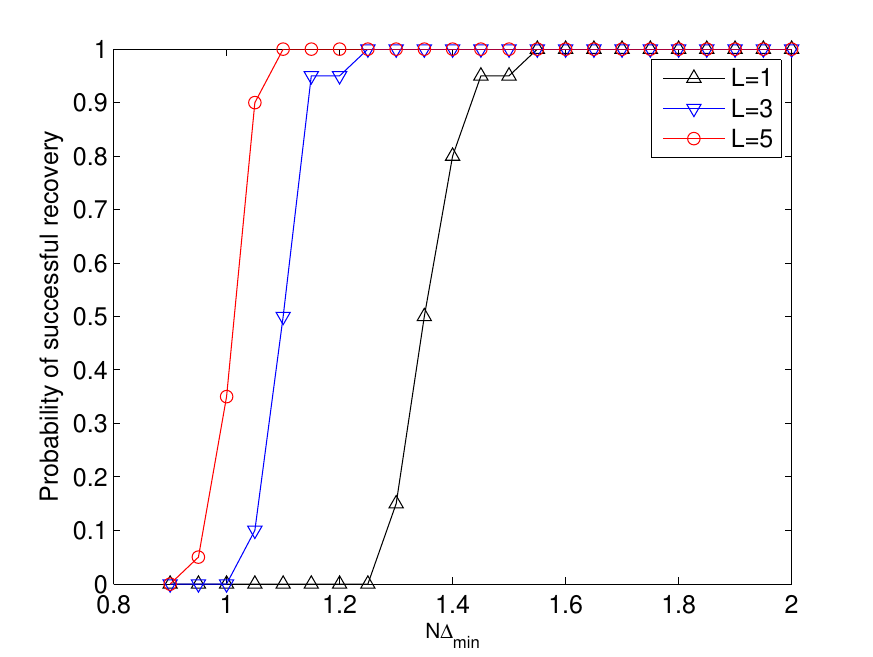}}
\centering
\caption{Frequency recovery results with respect to the number of measurement vectors $L$ in the case of full data and uncorrelated sources.} \label{Fig:Completedata1}
\end{figure}

We first consider uncorrelated sources, where the source signals $\m{S}=\mbra{s_{kt}}\in\bC^{K\times L}$ in (\ref{formu:observmodel1}) are drawn i.i.d. from a standard complex Gaussian distribution. Moreover, we consider the number of measurement vectors $L=1,3, \text{ and }5$. For each value of $\Delta_{\text{min}}$ and each type of frequencies, we carry out 20 Monte Carlo runs and calculate the success rate of frequency recovery. In each run, we generate $\cT$ and $\m{S}\in\bC^{K\times 5}$ and obtain the full data $\m{Y}^o$. For each value of $L$, we attempt to recover the frequencies using the proposed atomic norm method, implemented by SDPT3 \cite{toh1999sdpt3} in Matlab, based on the first $L$ columns of $\m{Y}^o$. The frequencies are considered to be successfully recovered if the root mean squared error (RMSE) is less than $10^{-8}$.

The simulation results are presented in Fig. \ref{Fig:Completedata1}, which verify the conclusion of Theorem \ref{thm:completedata} that the frequencies can be exactly recovered using the proposed atomic norm method under a frequency separation condition. When more measurement vectors are available, the recovery performance improves and it seems that a weaker frequency separation condition is sufficient to guarantee exact frequency recovery. By comparing Fig. \ref{Fig:Complete_L_UniformFreq} and Fig. \ref{Fig:Complete_L_NonuniformFreq}, it also can be seen that a stronger frequency separation condition is required in the case of equispaced frequencies where more frequencies are present and they are located more closely.

\begin{figure}
\centering
  \subfigure[Equispaced frequencies]{
    \label{Fig:Complete_CohSource_UniformFreq}
    \includegraphics[width=3in]{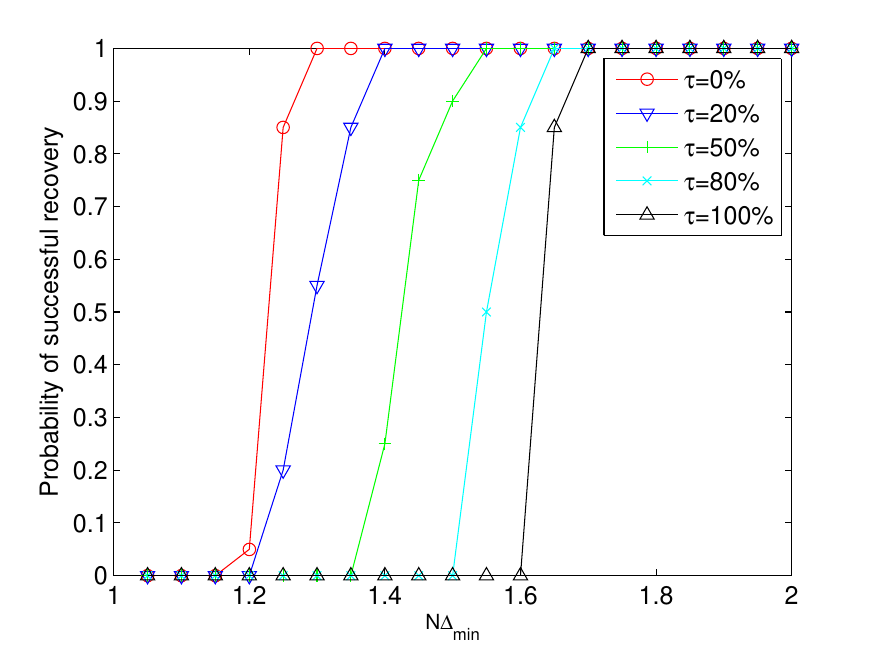}}
  \subfigure[Random frequencies]{
    \label{Fig:Complete_CohSource_NonuniformFreq}
    \includegraphics[width=3in]{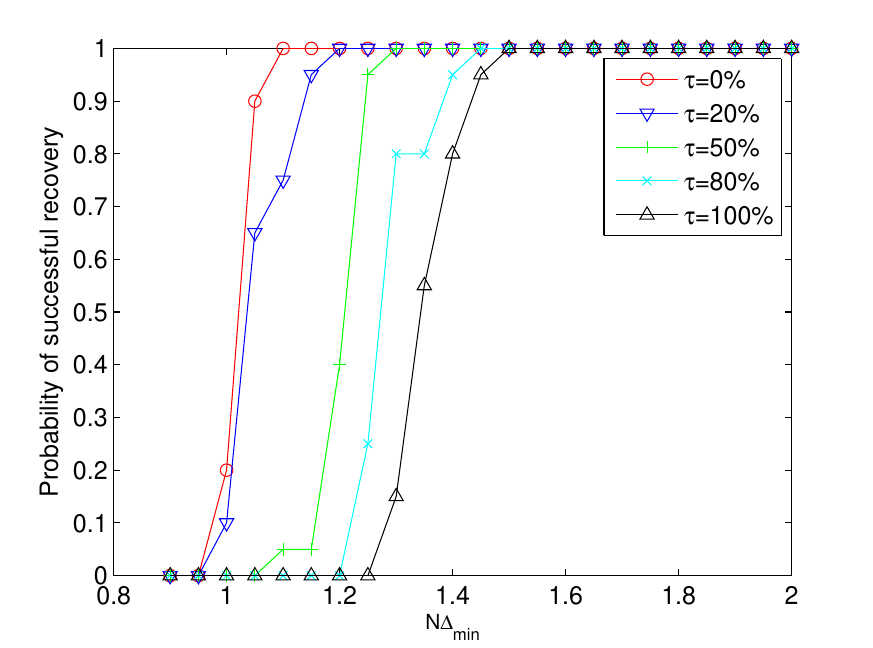}}
\centering
\caption{Frequency recovery results with respect to the percentage of coherent sources $\tau$ in the case of full data and coherent sources, with $L=5$.} \label{Fig:Completedata2}
\end{figure}

We next consider coherent sources. In this simulation, we fix $L=5$ and consider different percentages, denoted by $\tau$, of the $K$ source signals which are coherent (identical up to a scaling factor). It follows that $\tau=0\%$ refers to the case of uncorrelated sources considered previously. $\tau=100\%$ means that all the sources signals are coherent and the problem is equivalent to the SMV case. For each type of frequencies, we consider five values of $\tau$ ranging from $0\%$ to $100\%$ and calculate each success rate over 20 Monte Carlo runs.

Our simulation results are presented in Fig. \ref{Fig:Completedata2}. It can be seen that, as $\tau$ increases, the success rate decreases and a stronger frequency separation condition is required for exact frequency recovery. As $\tau$ equals the extreme value $100\%$, the curves of success rate approximately match those for $L=1$ in Fig. \ref{Fig:Completedata1}, verifying that taking MMVs does not necessarily improve the performance of frequency recovery.\footnote{The slight differences between the curves in Fig. \ref{Fig:Completedata1} and Fig. \ref{Fig:Completedata2} are partially caused by the fact that, to simulate coherence sources, the two datasets are generated slightly differently.}

Finally, we report the computational speed of the proposed atomic norm method. It takes about 11s to solve one SDP on average on a PC and the CPU times differ slightly for the three values of $L$. About 22 hours are used in total to produce the data generating Fig. \ref{Fig:Completedata1} and Fig. \ref{Fig:Completedata2}.


\subsection{Compressive Data}

In the compressive data case, we study the so-called phase transition phenomenon in the $\sbra{M,K}$ plane. In particular, we fix $N=128$, $L=5$ and $\Delta_{\text{min}}=1.2N^{-1}$, and study the performance of the proposed ANM method in signal and frequency recovery with different settings of the source signal. The frequency set $\cT$ is randomly generated with $\Delta_{\cT}\geq\Delta_{\text{min}}$ and $\abs{\cT}=K$ (differently from that in the last subsection, the process of adding frequencies is terminated as $\abs{\cT}=K$). In our simulation, we vary $M=8,12,\dots,128$ and at each $M$, $K=2,4,\dots,\min(M,84)$ since it is difficult to generate a set of frequencies with $K>84$ under the aforementioned frequency separation condition. In this simulation, we consider temporarily correlated sources. In particular, suppose that each row of $\m{S}$ has a Toeplitz covariance matrix $\m{R}\sbra{r}=\begin{bmatrix}1&r&\dots&r^4\\ r&1&\dots&r^3 \\ \vdots & \vdots & \ddots & \vdots\\ r^4 & r^3 & \dots & 1\\\end{bmatrix}\in\bR^{5\times5}$ (up to a positive scaling factor). Therefore, $r=0$ means that the source signals at different snapshots are uncorrelated while $r=\pm1$ means completely correlated and corresponds to the trivial case. We first generate $\m{S}_0$ from an i.i.d. standard complex Gaussian distribution and then let $\m{S}\sbra{r}=\m{S}_0\m{R}\sbra{r}^{\frac{1}{2}}$, where we consider $r=0,0.5, 0.9, 1$. For each combination $\sbra{M,K}$, we carry out 20 Monte Carlo runs and calculate the rate of successful recovery with respect to $r$. The recovery is considered successful if the relative RMSE of data recovery, measured by $\frobn{\m{Y}^*-\m{Y}^o}/\frobn{\m{Y}^o}$, is less than $10^{-8}$ and the RMSE of frequency recovery is less than $10^{-6}$, where $\m{Y}^*$ denotes the solution of $\m{Y}$.

\begin{figure*}
\centering
  \subfigure[$r=0$]{
    \label{Fig:pt00}
    \includegraphics[width=2.5in]{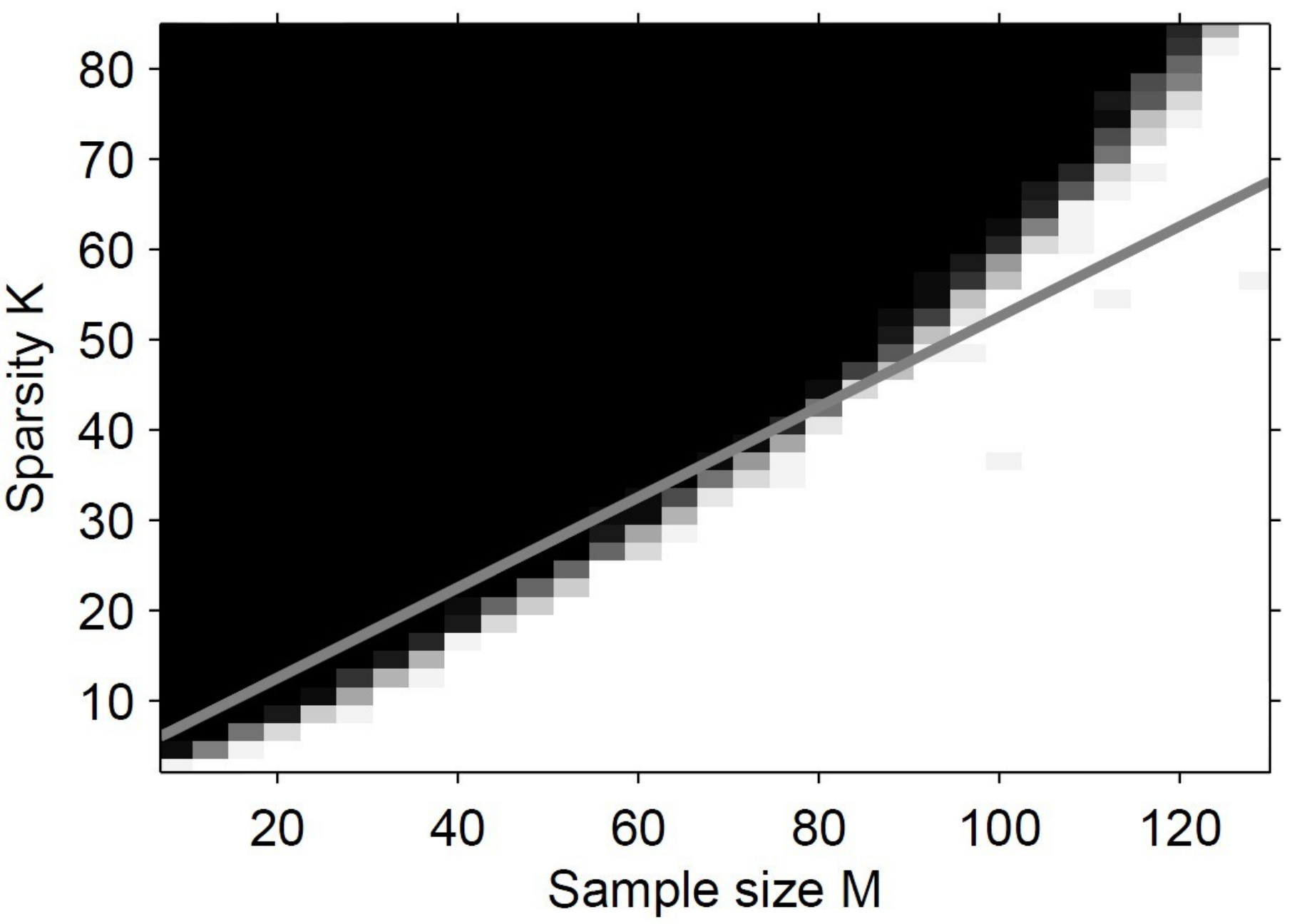}} %
  \subfigure[$r=0.5$]{
    \label{Fig:pt50}
    \includegraphics[width=2.5in]{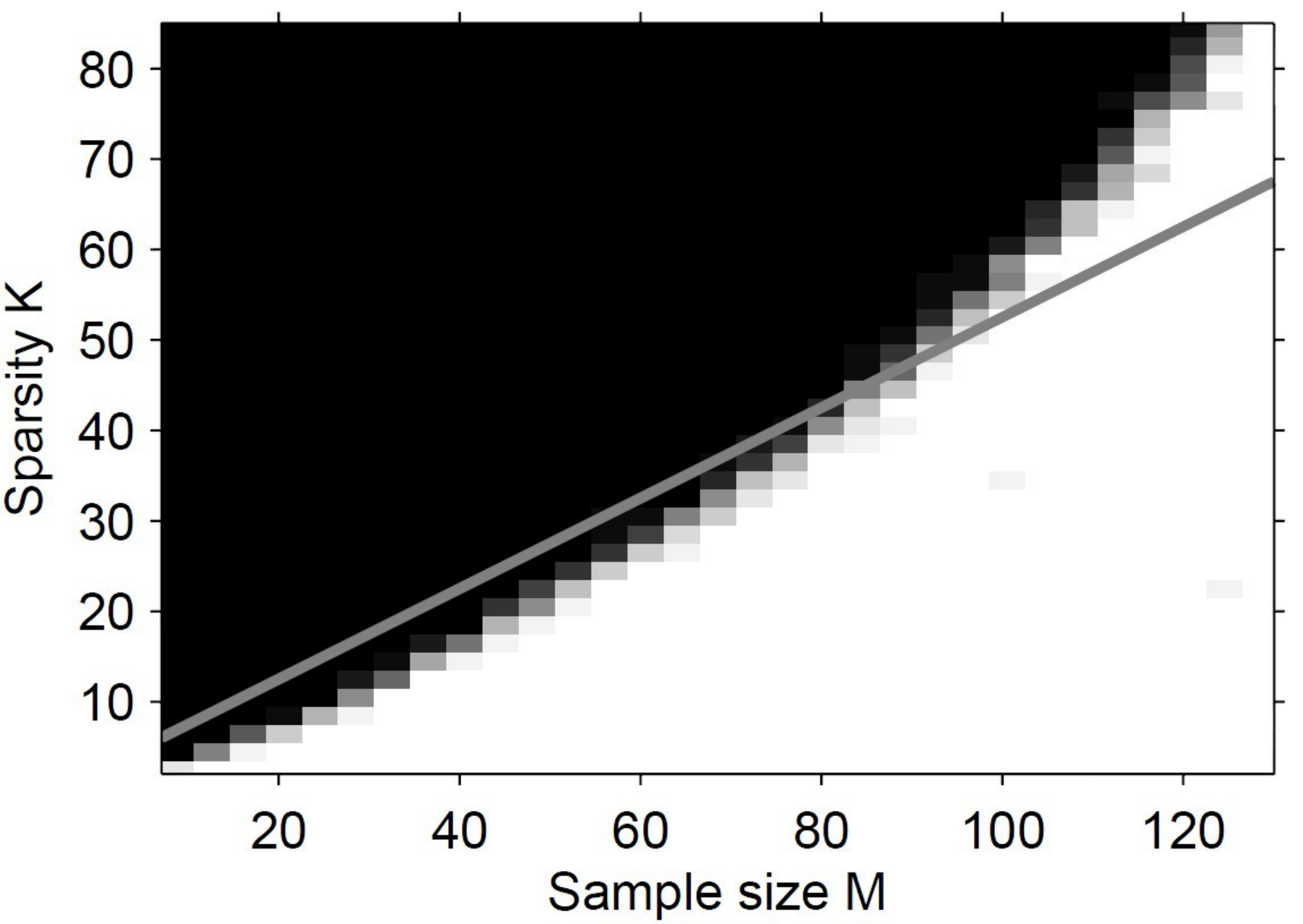}}
  \subfigure[$r=0.9$]{
    \label{Fig:pt90}
    \includegraphics[width=2.5in]{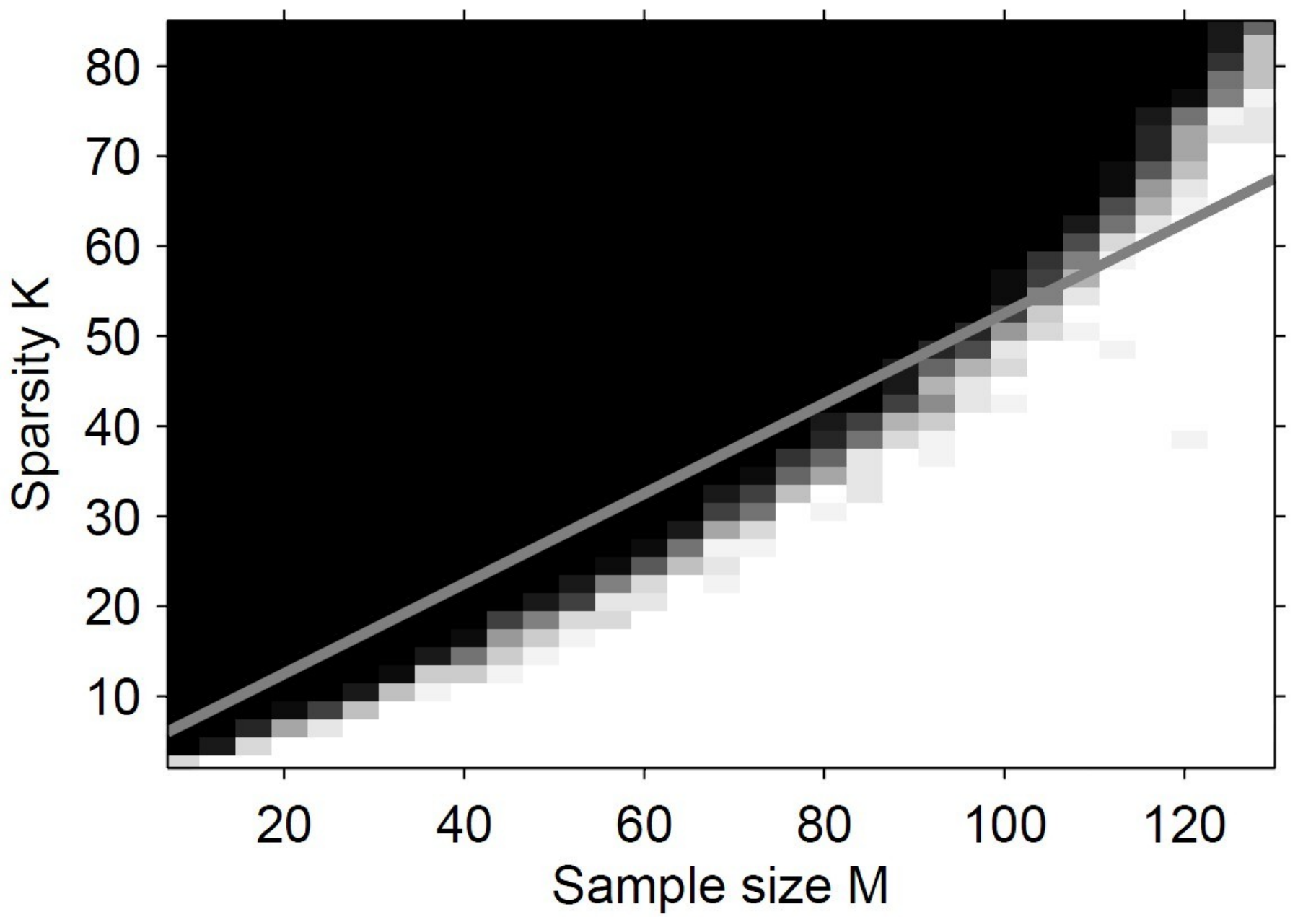}}%
  \subfigure[$r=1$]{
    \label{Fig:pt100}
    \includegraphics[width=2.5in]{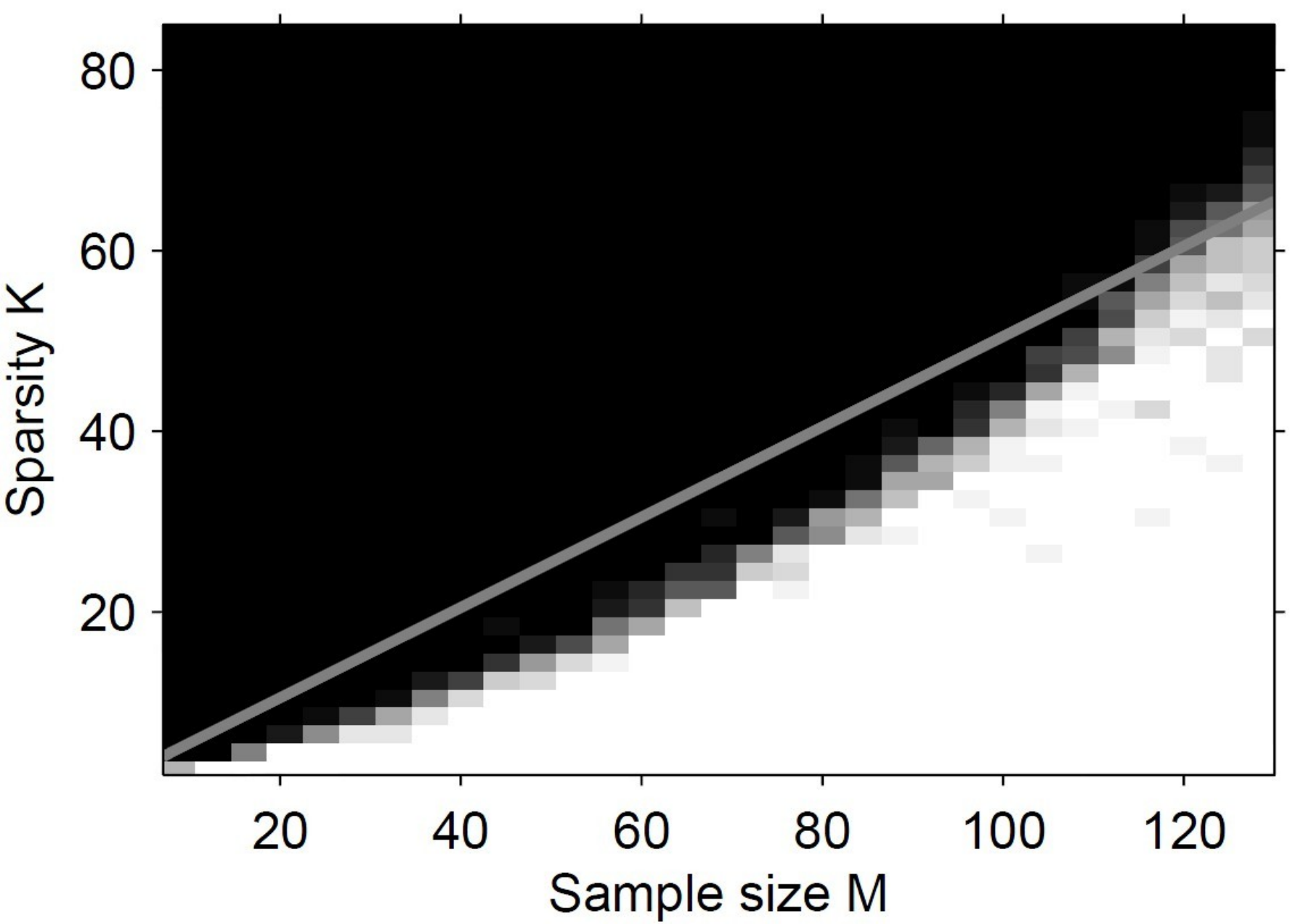}}
\centering
\caption{Phase transition results in the compressive data case with $N=128$ and $\Delta_{\text{min}}=1.2N^{-1}$. White means complete success and black means complete failure. The straight lines are $K=\frac{1}{2}\sbra{M+L}$ in (a)-(c) and $K=\frac{1}{2}\sbra{M+1}$ in (d).} \label{Fig:phasetransition}
\end{figure*}

The simulation results are presented in Fig. \ref{Fig:phasetransition}, where the phase transition phenomenon from perfect recovery to complete failure can be observed in each subfigure. It can be seen that more frequencies can be recovered when more samples are observed. When the correlation level of the MMVs, indicated by $r$, increases, the phase of successful recovery decreases. On the other hand, note that Fig. \ref{Fig:pt100} actually corresponds to the SMV case. By comparing Fig. \ref{Fig:pt100} and the other three subfigures, it can be seen that the frequency recovery performance can be greatly improved by taking MMVs, even in the presence of strong temporal correlations.

We also plot the line $K=\frac{1}{2}\sbra{M+L}$ in Figs. \ref{Fig:pt00}-\ref{Fig:pt90} and the line $K=\frac{1}{2}\sbra{M+1}$ in Fig. \ref{Fig:pt100} (straight gray lines) which are upper bounds of the sufficient condition in Theorem \ref{thm:AL0_guanrantee} for the atomic $\ell_0$ norm minimization (note that $\spark\sbra{\cA^1_{\m{\Omega}}}\leq M+1$). It can be seen that successful recoveries can be obtained even above these lines, indicating good performance of the proposed ANM method. It requires about 13s on average to solve one problem and almost 200 hours in total to generate the whole data set used in Fig. \ref{Fig:phasetransition}.

\subsection{The Noisy Case}

\begin{figure*}
\centering
  \subfigure[$L=1$]{
    \label{Fig:noisycase1}
    \includegraphics[width=2.2in]{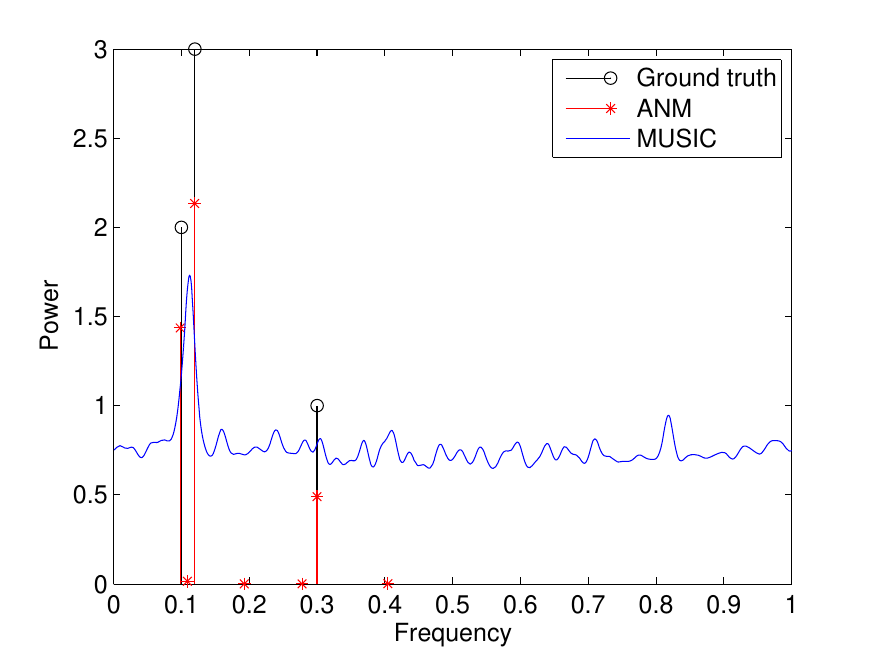}} %
  \subfigure[$L=5$, uncorrelated sources]{
    \label{Fig:noisycase2}
    \includegraphics[width=2.2in]{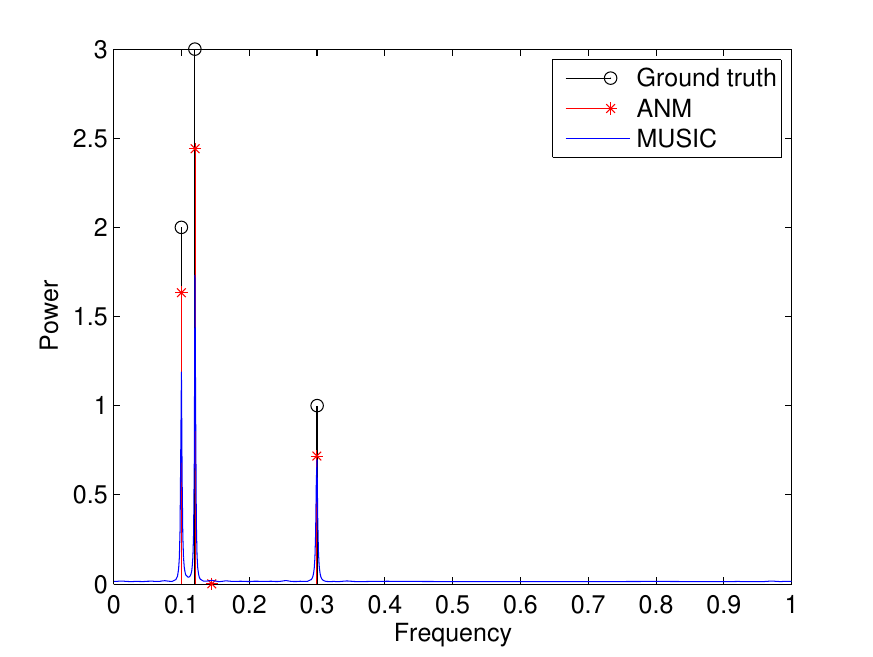}} %
  \subfigure[$L=5$, sources 1 and 3 are coherent]{
    \label{Fig:noisycase3}
    \includegraphics[width=2.2in]{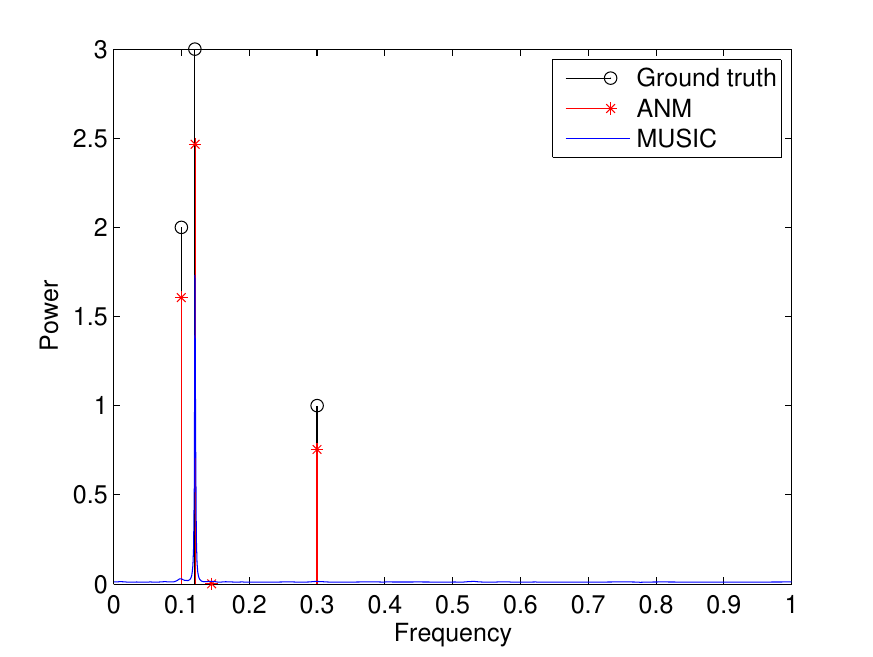}}
\centering
\caption{Frequency recovery/estimation results of ANM and MUSIC in the presence of noise, with (a) $L=1$, (b) $L=5$ and  uncorrelated sources, and (c) $L=5$ and coherent sources.} \label{Fig:noisycase}
\end{figure*}

While this paper has been focused on the noiseless case, we provide a simple simulation to illustrate the performance of the proposed method in the practical noisy case. We consider $N=50$, $M=20$ with $\m{\Omega}$ randomly generated, $K=3$ sources with frequencies of $0.1$, $0.12$ and $0.3$ and powers of 2, 3 and 1 respectively, and $L=5$. The source signals of each source are generated with constant amplitude and random phases. Complex white Gaussian noise is added to the measurements with noise variance $\sigma^2=0.1$. We propose to denoise the observed noisy signal $\m{Y}^o_{\m{\Omega}}$ and recover the frequency components by solving the following optimization problem:
\equ{\min_{\m{Y}} \norm{\m{Y}}_{\cA}, \st \frobn{\m{Y}_{\m{\Omega}}-\m{Y}^o_{\m{\Omega}}}^2\leq\eta^2, }
where $\eta^2$, set to $\sbra{ML+2\sqrt{ML}}\sigma^2$ (mean + twice standard deviation), bounds the noise energy from above with large probability. The spectral MUSIC method is also considered for comparison. Note that MUSIC estimates the frequencies from the sample covariance, while the proposed ANM method carries out covariance fitting by exploiting its structures. While the proposed method requires the noise level, MUSIC needs the source number $K$.

Typical simulation results of one Monte Carlo run are presented in Fig. \ref{Fig:noisycase}. The SMV case is studied in Fig. \ref{Fig:noisycase1} where only the first measurement vector is used for frequency recovery. It is shown that the three frequency components are correctly identified using the ANM method while MUSIC fails. The MMV case is studied in Fig. \ref{Fig:noisycase2} with uncorrelated sources, where both ANM and MUSIC succeed to identify the three frequency components. The case of coherent sources is presented in Fig. \ref{Fig:noisycase3}, where source 3 in Fig. \ref{Fig:noisycase2} is modified such that it is coherent with source 1. MUSIC fails to detect the two coherent sources as expected while the proposed method still performs well. It is shown in all the three subfigures that spurious frequency components can be present using the ANM method. But their powers are low. To be specific, the spurious components have about $0.4\%$ of the total powers in Fig. \ref{Fig:noisycase1}, and this number is on the order of $10^{-6}$ in the latter two subfigures. While these numerical results imply that the proposed method is robust to noise, a theoretical analysis will be investigated in future studies. The proposed method needs about $1.5$s in each scenario.

\section{Conclusion} \label{sec:conclusion}
In this paper we studied the JSFR problem by exploiting the joint sparsity in the MMVs. We proposed an atomic $\ell_0$ norm approach and showed the advantage of MMVs. We also proposed an atomic norm approach that can be efficiently solved by semidefinite programming and studied its theoretical guarantees for frequency recovery. These results extend the existing ones either from the SMV to the MMV case or from the discrete to the continuous frequency setting. We also discussed the connections between the proposed approaches and conventional subspace methods as well as the recent grid-based and gridless sparse techniques. Though the worst case analysis we provided for the atomic norm approach does not indicate performance gains in the presence of MMVs, simulation results indeed imply that when the source signals are located at general positions the number of required measurements can be reduced and/or the frequency separation condition can be relaxed. This average case analysis should be investigated in future studies under stronger assumptions.

\section*{Acknowledgement}
Z. Yang would like to thank Dr. Gongguo Tang of Colorado School of Mines, USA for helpful discussions.



\end{document}